%% file: quantum-template.tex
\documentclass[a4paper,onecolumn,11pt,accepted=2026-06-08]{quantumarticle}
\pdfoutput=1
\usepackage[utf8]{inputenc}
\usepackage[english]{babel}
\usepackage[T1]{fontenc}
\usepackage{amsmath}
\usepackage{hyperref}
\usepackage{tikz}
\usepackage{lipsum}
\usepackage{graphicx} 
\usepackage{mathtools,amsmath,amsfonts,amssymb,amsthm}
\usepackage{braket} 
\usepackage{ascmac}
\usepackage{multirow}

\theoremstyle{definition}
\newtheorem{theorem}{Theorem}
\newtheorem{definition}{Definition}
\newtheorem{corollary}{Corollary}

\newtheorem{lemma}{Lemma}
\newtheorem{remark}{Remark}
\newtheorem{claim}{Claim}
\newtheorem{prop}{Proposition}

\newcommand{\QMA}{\textup{QMA}}

\newcommand{\DQC}{\textup{DQC$_1$}}

\newcommand{\nats}{{\mathbb N}}

\begin{document}

\title{Quantum Walks on Simplicial Complexes and Harmonic Homology: Application to Topological Data Analysis with Superpolynomial Speedups}

\author{Ryu Hayakawa}
\affiliation{The Hakubi Center/Yukawa Institute for Theoretical Physics, Kyoto University, Japan}
\author{Kuo-Chin Chen}
\affiliation{Foxconn Research, Taipei, Taiwan}
\author{Min-Hsiu Hsieh}
\affiliation{Foxconn Research, Taipei, Taiwan}
\maketitle

\begin{abstract}
Incorporating higher-order interactions in information processing enables us to build more accurate models, gain deeper insights into complex systems, and address real-world challenges more effectively. However, existing methods, such as random walks on oriented simplices and homology, which capture these interactions, are not known to be efficient. This work investigates whether quantum walks on simplicial complexes exhibit quantum advantages. 
We introduce a novel quantum walk that encodes the combinatorial Laplacian, a key mathematical object whose spectral properties reflect the topology of the underlying simplicial complex. Furthermore, we construct a unitary encoding that projects onto the kernel of the Laplacian, representing the space of harmonic cycles in the complex's homology. Combined with the efficient construction of quantum walk unitaries for clique complexes that we present, this paves the way for utilizing quantum walks to explore higher-order interactions within topological structures. Specifically, our construction requires $\mathcal{O}(n^3\log(1/\epsilon)/\lambda_k)$ gates, where $n$ is the number of vertices, $\lambda_k$ is the smallest non-zero eigenvalue of the combinatorial Laplacian, and $\epsilon$ is the error of the projection, for the implementation of the projector onto the harmonic cycles. 
{Thus, our results indicate apparent superpolynomial quantum speedup with quantum 
walks without relying on quantum oracles for large datasets, 
provided that the spectral gap of the combinatorial Laplacian 
is inverse-polynomially bounded and efficient sampling of simplices
is available.}

Crucially, the walk operates on a state space encompassing both positively and negatively oriented simplices, effectively doubling its size compared to unoriented approaches. Through coherent interference of these paired simplices, we are able to successfully encode the combinatorial Laplacian, which would otherwise be impossible. This observation constitutes our major technical contribution. 
We also extend the framework by constructing variant quantum walks. These variants enable us to: (1) estimate the normalized persistent Betti numbers, capturing topological information throughout a deformation process, (2) verify a specific QMA$_1$-hard problem related to the homology of clique complexes, showcasing potential applications in computational complexity theory, {and (3) solve the high-dimensional discrete 
Dirichlet problem (HDDP), a generalization of the classical 
discrete Dirichlet problem on graphs to the setting of 
simplicial complexes, with an apparent superpolynomial speedup over the best
known classical algorithm.}
\end{abstract}

\section{Introduction}

\label{sec:randomwalks}

Graph-based information processing has seen extensive success over several decades. However, a fundamental limitation of graphs is their inability to model beyond pairwise interactions. Recent evidence has highlighted the importance of higher-order interactions, where links can connect more than two nodes. Consequently, higher-order networks have emerged as a crucial frontier in complex network science~\cite{bick2023higher, battiston2020networks}. A prominent model in this domain is simplicial complexes, which are natural generalizations of graphs to higher-order structures. 

Just as random walks on graphs have been important tools to solve graph problems, random walks on simplicial complexes can play {an important} role in analyzing higher-order networks. Graph Laplacian dynamics and homology are significant concepts in graph-based information processing. Therefore, it is natural to explore the theoretical connection between higher-order random walks and the topological properties of underlying simplicial complexes. Recent breakthroughs have revealed close relationships between the topological properties of simplicial complexes and random walks on them~\cite{parzanchevski2017simplicial, mukherjee2016random, schaub2020random, rosenthal2014simplicial, eidi2023irreducibility}.

The concept of orientation becomes crucial when relating random walks on simplicial complexes to their homology. While $0$-simplices are merely vertices and lack a natural orientation, for $k\geq 1$, there exists a clear distinction between positive and negative orientations of simplices. For instance, in the case of $1$-simplices (edges), positive and negative orientations (directions) can be discerned. Analogously, similar notions of positive and negative orientation can be extended to higher-dimensional simplices. Random walks described in~\cite{parzanchevski2017simplicial, mukherjee2016random, schaub2020random} are conducted on such oriented simplicial complexes. In essence, these approaches treat identical simplices with different orientations as distinct states. By exploring random walks on such oriented simplicial complexes, one can examine the {\it difference} between the probabilities of being in a positive simplex and its corresponding negative simplex. This process, defined through the difference of probabilities, is termed the ``expectation process"~\cite{parzanchevski2017simplicial, mukherjee2016random}. It has been demonstrated in~\cite{mukherjee2016random} that the normalized version of the expectation process exhibits a rigorous connection with the projector onto the harmonic homology group of the underlying simplicial complex. 
It is shown in~\cite{mukherjee2016random}
that performing random walks from some initial simplex and executing the normalized expectation process is equivalent to approximately applying the projector onto  the harmonic homology group to that initial state. 
It is shown as long as the spectral gap of the combinatorial Laplacian operator is inverse-polynomially large, the normalized expectation process can implement the projector with high precision with polynomial time steps. 
However, since the normalized expectation process differs from the conventional random walk itself, it remains open whether the process can be simulated efficiently, especially in (extremely) high-dimensional cases.

Analysis of higher-order networks has recently gained attention in the quantum computing literature. One significant focus has been on efficiently performing topological data analysis (TDA)~\cite{edelsbrunner2022computational, dey2022computational} in the high-dimensional regime using quantum computing.
In TDA, input data is transformed into a family of simplicial complexes (referred to as filtration), and the goal is to analyze how topological invariants vary across the filtered simplicial complexes. The primary topological invariants of interest in TDA are known as persistent Betti numbers. While classical TDA has become a popular field of data analysis, classical algorithms face limitations in analyzing high-dimensional persistent Betti numbers. 
This is because the number of simplices can grow superpolynomially in the number of vertices.

As exact computation is even intractable for quantum computers \footnote{Indeed, exactly calculating the Betti numbers is known to be \#P-hard for clique complexes~\cite{schmidhuber2023complexity} and more precisely, it is \#\QMA$_1$-hard~\cite{crichigno2022clique}.}, in quantum TDA (QTDA), a relaxation of the condition of estimating the persistent Betti numbers is usually considered, namely the estimation of the normalized Betti numbers with inverse-polynomial additive precision. 
The first QTDA algorithm was proposed by Lloyd, Garnerone, and Zanardi~\cite{lloyd2016quantum}. 
The estimation of the normalized Betti number has been extended to the estimation of the normalized persistent Betti numbers~\cite{hayakawa2022quantum, mcardle2022streamlined}. QTDA algorithms exhibit superpolynomial speedup compared to the best-known classical algorithms, because there is no known efficient classical algorithm for solving the NBE problem with inverse polynomial precision in $n$.

Utilizing Quantum Walks (QW) based algorithms has shown significant speed enhancements compared to classical algorithms across various scenarios, and demonstrated polynomial speedup over classical algorithms \cite{aaronson2004quantum,szegedy2004quantum,buhrman2006quantum,magniez2007quantum,somma2008quantum,apers2019quantum,ambainis2020quadratic}.
However, there are a few problems that exhibit exponential speedup.
The first QW algorithm exhibiting exponential speedup was proposed in \cite{childs2003exponential}, which utilized continuous-time quantum walks to traverse the welded-tree graph with height $n$ by using $O(n^6)$ queries, demonstrating that no classical algorithm can match its efficiency. 
This problem can be solved by a multidimensional QW algorithm with a better complexity\cite{jeffery2023multidimensional}, and can be extended to a broad class of hierarchical graphs\cite{balasubramanian2023exponential}.
Although these results show an exponential advantage over the classical algorithms, the advantage holds in the {\it query complexity} with oracle access to exponentially large data. 
Therefore, it does not mean that these QW algorithms demonstrate superpolynomial speedup in the realistic computation time.

There are a few existing quantum walk results considering higher dimensional space than the graphs. 
The first quantum walk on simplicial complexes is proposed by Matsue, Ogurisu, and Segawa in \cite{matsue2016quantum}. They extend Szegedy’s QW on simplicial complexes, and these quantum walks possess linear spreading and localization as the Grover walk. 
Later on, they proposed a QW algorithm that can find a marked simplex polynomially faster than the classical counterpart \cite{matsue2018quantum}. 
However, these results do not show a superpolynomial quantum advantage.
{Another line of work introduces Grover-type quantum walks constructed using coin operators together with shift operators that encode the orientation of each simplex \cite{luo2018up}. The corresponding discriminant operators are closely related to the up and down combinatorial Laplacians, revealing a deep connection between the quantum walk dynamics and the underlying geometry of the simplicial complex. However, these results are primarily structural and do not directly address dynamical aspects such as mixing behavior.}
In addition, it has been numerically observed that quantum walks on simplicial complexes reflect the topology of simplicial complexes~\cite{matsue2016quantum}. However, the connection between the quantum walks and topology of simplicial complexes has not been rigorously proven.

Even though there is no known efficient way to simulate the normalized expectation process in arbitrary dimensions, 
it indeed paves new applications for signal processing and learning tasks on simplicial complexes in constant dimensions such as the random walks on edges or triangles. 
For example, random walks on oriented edges have been applied to the task of label propagation on edges in~\cite{mukherjee2016random}. 
In~\cite{schaub2020random}, similar higher-order random walks are applied to the task of edge flow estimation and Pagerank vector on simplicial complexes. For a review, see~\cite{schaub2021signal,battiston2020networks}. 

\subsection{Main results}

In this paper, we show a rigorous relationship between quantum walks and homology for the first time. 
Our main result concerns the implementation of projectors onto subspaces that are related to the topology of simplicial complexes with quantum walks. 
Let us focus on the $k$-th {\it harmonic homology group}, denoted by $\mathcal{H}_k$.
For a trivial $\mathcal{H}_k$, the simplicial complex has no ``$k$-dimensional holes''. Conversely, the dimension of $\mathcal{H}_k$ reflects the number of such holes, making it highly relevant to the complex's topology.
We can relate $\mathcal{H}_k$ to a positive semi-definite hermitian operator called the combinatorial Laplacian $\Delta_k$. 
It is known that 
$\mathcal{H}_k=\ker (\Delta_k)$.  
Our main result can be stated as follows. 

{
\begin{theorem}[Quantum walk based projection onto harmonics (informal statement of Theorem
\ref{thm:main})]
\label{thm:informal_projectors}
Let $X$ be a simplicial complex over $n$ vertices such that 
$\lambda_k>1/poly(n)$, where $\lambda_k$ is the smallest non-zero eigenvalue of $\Delta_k$. 
Then, we can implement a unitary encoding of an operator that is inverse-exponentially close to the projector onto $\mathcal{H}_k$ with $poly(n)$ use of $U^{h}$ and other elementary quantum gates. 
Here, $U^h$ is a harmonic quantum walk unitary operator in which a Markov transition matrix $P$ on oriented $k$-simplices of eq.~\eqref{eq:Markov_chain} is encoded using $\alpha$-ancilla qubits i.e., $U^h$ is a unitary s.t.
$$  \left(\Pi^\pm_{k}\otimes\bra{0^{\alpha}}\right)U^{h} \left(\Pi^\pm_{k}\otimes \ket{0^{\alpha}}\right)= \Pi^\pm_{k}P\Pi^\pm_{k},
$$
where $\Pi^\pm_{k}$ is a projector onto the space spanned by positively and negatively oriented $k$-simplices.
\end{theorem}
}

In addition to the implementation of the projector onto $\mathcal{H}_k$, we implement quantum encodings of projectors onto other subspaces related to the topology 
based on variants of quantum walks on simplicial complexes ({See the formal statement of Theorem~\ref{thm:main}}).
These projectors play essential roles in the application of quantum walks to the estimation of the normalized {\it persistent} Betti numbers, which is an indispensable topological invariant in TDA. 
In total, we construct three quantum walk algorithms and they are referred to as up, down, and harmonic quantum walks.

Our second contribution is that we provide an explicit construction of quantum walk unitaries in Section~\ref{sec:clique_construction}. Especially, we show that we can efficiently construct quantum walks unitaries for {\it clique complexes} given access to the graph adjacency matrix, which is (exponentially) succinct compared to the size of the simplicial complex. 

\begin{theorem}[Informal statement of Theorem~\ref{theorem:clique}]
\label{thm:informal_clique}
Let $X$ be a vertex-weighted clique complex over $n$ vertices. 
Then, we can construct unitaries $U^{up}$, $U^{down}$, and $U^{h}$ 
that are unitary encodings of up, down, and harmonic Markov transition matrices with error $\epsilon$ using $\alpha=\mathcal{O}(n)$ ancilla qubits
with ${\mathcal{O}}(n^2\log(1/\epsilon))$-number of gates, respectively. 
\end{theorem}

Our results on the encoding of Laplacians, implementation of projectors, and construction of quantum walk unitaries are summarized in Table~\ref{tab:results}.

\begin{table}
    \centering
    \begin{tabular}{ccccc} 
    \multirow{2}{*}{QW}&\multirow{2}{*}{Quantization of}&\multirow{2}{*}{Encoding}&$\epsilon$-approx. projector (Thm.~\ref{thm:main})& \multirow{2}{*}{Efficient construction}\\
         &   &  & \#-of use of QW unitaries & \\ \hline
         \multirow{2}{*}{$U^{up}$}&  $P^{up}$ &  $\Delta^{up}_k$ &  $Z^k$, $B_k$& 
         \\
 & (Def.~\ref{def:down_random_walk})& (Prop.~\ref{prop:unitaryencoding_up})& $\mathcal{O}(K^{up}\log(1/\epsilon)/\lambda_k^{up})$ & Clique complexes\\ \cline{1-4}
         \multirow{2}{*}{$U^{down}$}&  $P^{down} $ &  $\Delta^{down}_k$ &  $Z_k$, $B^k$& $\mathcal{O}(n^2)$-gates\\
 & (Def.~\ref{def:up_random_walk})& (Prop.~\ref{prop:unitaryencoding_down})&$\mathcal{O}(K^{down}\log(1/\epsilon)/\lambda_k^{down})$ & (Thm.~\ref{theorem:clique})\\ \cline{1-4}
         \multirow{2}{*}{$U^h$}&  $P$ &  $\Delta_k$ &  $\mathcal{H}_k$&  \\
 & (Def.~\ref{def:harmonic_random_walk})& (Prop~\ref{prop:unitaryencoding})& $\mathcal{O}(K\log(1/\epsilon)/\lambda_k)$  &\\ \hline
    \end{tabular}
    \caption{Our results on the encoding of Laplacians, implementation of projectors, and construction of quantum walk unitaries. Here, $n$ is the number of vertices. $Z^k$, $B_k$, $Z_k$, $B^k$, $\mathcal{H}_k$ are $k$-cocycle, $k$-boundary, $k$-cycle, $k$-coboundary and $k$-harmonics. $\lambda_k^{up}$, $\lambda_k^{down}$, $\lambda_k$ are the smallest non-zero eigenvalue of $\Delta_k^{up}$, $\Delta_k^{down}$ and $\Delta_k$. 
    $K^{up}$, $K^{down}$, $K$ are normalization factors of $\mathcal{O}(\|\Delta_k^{up}\|)$, $\mathcal{O}(\|\Delta_k^{down}\|)$ and $\mathcal{O}(\|\Delta_k\|)$.
    }
    \label{tab:results}
\end{table}

The presented Theorems~\ref{thm:informal_projectors} and~\ref{thm:informal_clique}
lead to new applications of quantum walks to higher-order networks because they clarify the conditions for the efficient implementations of the projectors and the efficient implementations of quantum walk unitaries for clique complexes. 
We demonstrate the usefulness of our quantum walks through {four} applications:

\begin{itemize}
    \item \textbf{Normalized Betti number estimation:} First, we show that we can efficiently estimate the normalized $k$-Betti numbers of clique complexes of arbitrary dimensions based on harmonic quantum walks if we can sample from the uniform distribution of $k$-simplices and the combinatorial Laplacian has an inverse-polynomial spectral gap (Theorem~\ref{thm:Betti_numbers}).
    \item \textbf{Normalized persistent Betti number estimation:} Second, we show that we can efficiently estimate the normalized $k$-persistent Betti numbers of clique complexes combining up quantum walks and down quantum walks (Theorem~\ref{thm:persisntent_Betti_numbers}). 
\item 
    {\textbf{High-Dimensional Discrete Dirichlet Problem (HDDP):} 
    Third, we apply our quantum walk framework to the HDDP~\cite{rosenthal2014simplicial}, 
    a higher-dimensional generalization of the discrete Dirichlet problem from graphs to 
    simplicial complexes. Given a partition of $k$-simplices into boundary and interior, 
    and a boundary cochain, we construct a quantum algorithm that outputs 
    a $k$-cochain satisfying the Dirichlet condition 
    (Theorem~\ref{thm:dirichlet}).
    }
    \item \textbf{Verification of the promise clique homology problem:} Fourth, we show we can efficiently verify the promise clique homology problem of~\cite{king2023promise} i.e., we can efficiently verify if a state sent from the prover is in the harmonic homology group or not (Theorem~\ref{thm:promise_clique_homology}). 
\end{itemize}
{The first two applications are motivated by TDA.
{The HDDP arises naturally in the study of harmonic 
$k$-forms with prescribed boundary conditions, connecting spectral geometry, Hodge theory, and combinatorial topology, and has many potential applications in higher-order networks.  As far as we know, our algorithm is the first quantum algorithm for HDDP.}
The fourth application is motivated by potentially providing new interactive protocols based on quantum walks such as the verification of quantum computing.}
These results also open the possibility of demonstrating superpolynomial speedup using quantum walks for practical problems on simplicial complexes. 
We discuss the computational complexity of the problems in our applications in Section~\ref{sec:related_works} and discuss why these problems might not be tractable for classical algorithms. 
Notably, unlike the previous superpolynomial speedup results of quantum walks for graph problems~\cite{childs2003exponential,balasubramanian2023exponential,jeffery2023multidimensional}, our superpolynomial speedup result does not rely on the query access to exponentially large matrices. 
This is because clique complexes can be described with graphs and we can efficiently construct quantum walk unitaries clique complexes from the graph adjacency matrices. 
As far as we know, this is the first example of superpolynomial speedup for graph or topological problems in the computational time rather than the query complexity based on quantum walks. 
{We expect that our quantum walks on simplicial complexes will find more applications in high-dimensional signal processing and machine learning tasks as is in the case for the random walks on simplicial complexes in constant dimensions.}

{
In summary, our main contribution is the discovery of a novel connection between quantum walks and the topology of simplicial complexes, inspired by the classical expectation process. 
This opens new applications of quantum walks 
for higher-order networks 
with potential superpolynomial speedup without relying on quantum oracles to exponentially large datasets, unlike previous claims of exponential speedups based on quantum walks~\cite{childs2003exponential,balasubramanian2023exponential,jeffery2023multidimensional}.
Compared to previous general quantum algorithms for topological data analysis~\cite{lloyd2016quantum,hayakawa2022quantum,mcardle2022streamlined,berry2024analyzing}, 
it is not that our algorithm provides improvement over them. 
However, our work provides a clearer picture of dynamics with quantum walks and connects the research of quantum TDA to the larger field of random walks and quantum walks. More discussions with previous literature can be found in Section~\ref{sec:related_works}.
}

\subsection{Technical Contributions}

No previous works have shown a relation between the homology of simplicial complexes and quantum walks. 
We identified that this obscurity originated from the lack of the proper treatment of orientations of simplicial complexes in quantum walks. 
It is the treatment of orientations of simplicial complexes that enabled us to encode the topological Laplacian operators into quantum walks and moreover, enabled us to implement the projectors onto the topological subspaces. Let us start by first illustrating how we represent the oriented simplicial complexes with quantum states.

\paragraph{Quantum representation of oriented simplices.}

{As we perform quantum walks on oriented simplicial complexes, how we represent such objects as quantum states comes up as the first technical challenge.}
The oriented $k$-simplices can be represented as $[v_0,v_1,...,v_k]$. 
The oriented $k$-simplices obtained by even (odd) permutation of vertices are said to have the same (opposite) orientation.
Throughout the paper, we fix the positively oriented simplices by the ordering $[v_0,v_1,...,v_k]$ with $v_0<v_1<...<v_k$. 
The unoriented simplices can be described by $n$-bit strings with Hamming weight $k+1$ where the positions of 1's correspond to the vertices of simplices. 
We represent each of oriented $k$-simplices $\sigma$ with $n+1$-qubit computational basis states where the first $n$ qubits correspond to the membership of the vertices and the last $1$ qubit represents the orientation i.e.,  
$$\ket{\sigma}=
\begin{cases}
    \ket{x_\sigma}\otimes \ket{0} & \text{if\ } \sigma\ \text{is\ positively\ oriented,}  \\
    \ket{x_\sigma}\otimes \ket{1} & \text{if\ } \sigma\ \text{is\ negatively\ oriented,} \\ 
\end{cases}$$
where $x_\sigma$ is the $n$-bit string with Hamming weight $k+1$.\footnote{We actually use $n+2$-qubit state to represent an additional special state called the absorbing state. Here, we describe without the absorbing state for simplicity.}
In Figure~\ref{fig:representation}, we show an example of quantum representation of a positively oriented $2$-simplex over $5$ vertices. 
\begin{figure}
    \centering
    \includegraphics[width=0.6\linewidth]{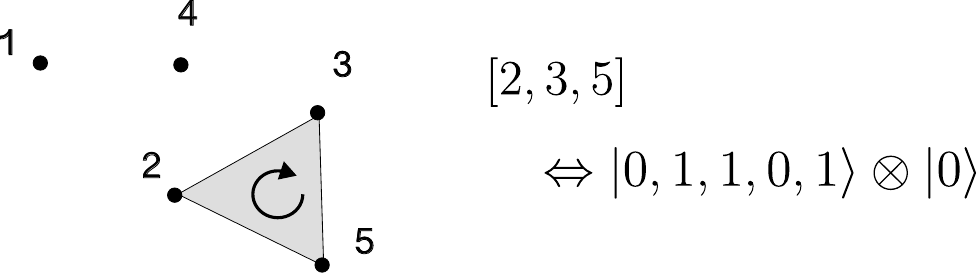}
    \caption{An example of quantum representation of oriented $2$-simplex (triangle) over $5$ vertices.}
    \label{fig:representation}
\end{figure}

\paragraph{Encoding of up, down, and full Laplacians into quantum walks.}

{The next technical challenge is how we can define quantum walks on simplicial complexes in which we can encode the Laplacians.} Unlike graphs, there are two different notions for adjacency for simplicial complexes: up-adjacency and down-adjacency. 
Based on the up-adjacency and down-adjacency, we can define three kinds of random walks on simplicial complexes (up, down, and harmonic random walks). Formal definitions can be seen in Definitions~\ref{def:up_random_walk},~\ref{def:down_random_walk},~\ref{def:harmonic_random_walk} in Section~\ref{sec:main}.  
Examples of up, down, and harmonic walks on simplicial complexes are shown in Figure~\ref{fig:up_walk}. 
For simplicity, we consider unweighted simplicial complexes here.

\begin{figure}
    \centering
    \includegraphics[width=13cm]{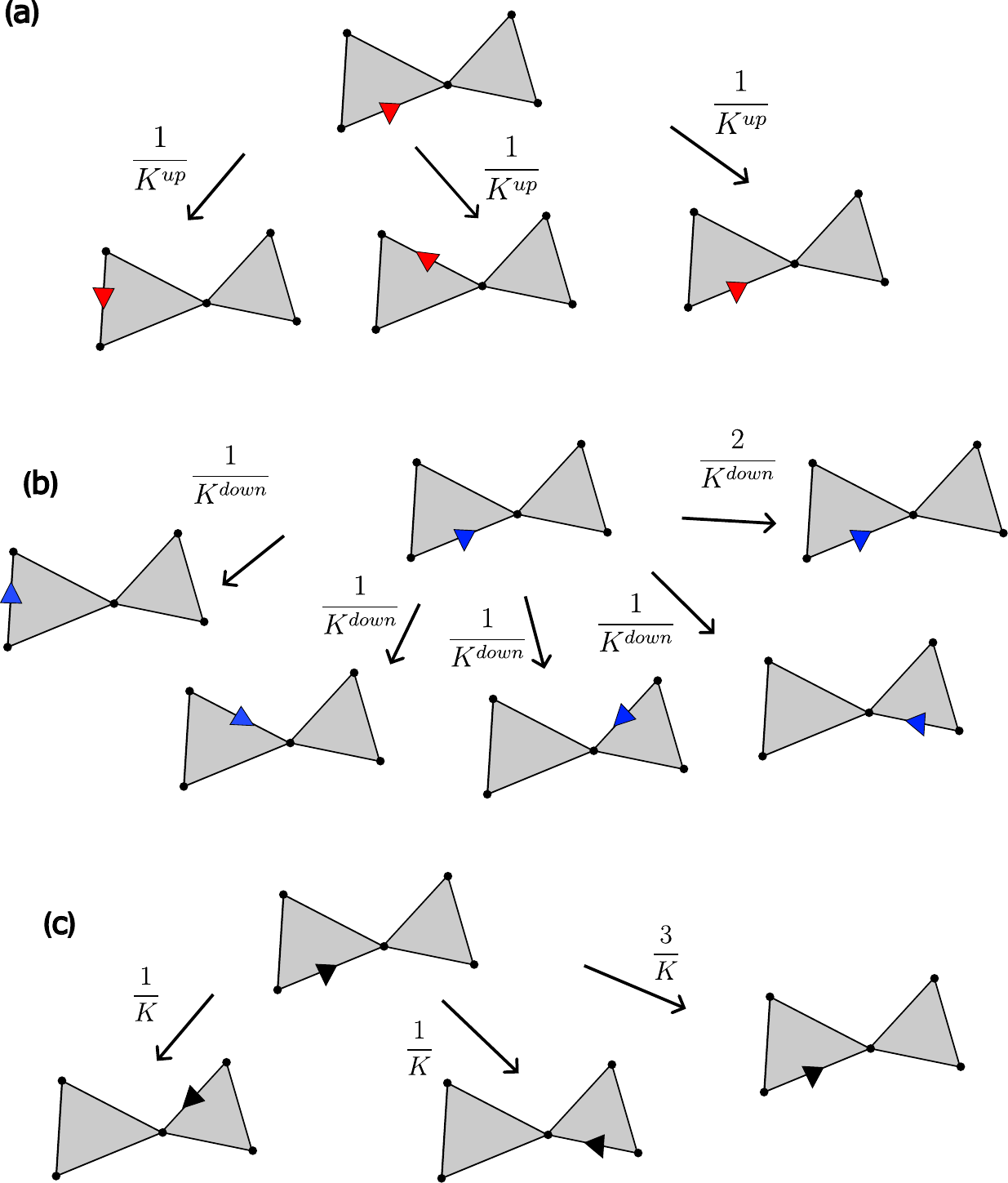}
    \caption{
    Examples of random walks on unweighted $1$-simplices (edges). 
    (a) An example of a step of up random walk on $1$-simplices. 
    $K^{up}$ is some normalization factor. 
    An edge transits to other edges that share common upper simplices (triangles) with equal probabilities. 
    The orientation is shown by the arrow. 
    The laziness of the up walk is proportional to the up-degree (the number of triangles that an edge has).
    (b) An example of a step of down random walk on $1$-simplices. $K^{down}$ is some normalization factor. 
    A state moves to other edges that share common lower-simplices (vertices) and stays on the same edge with probability proportional to the down degree (the number of vertices of edges i.e., 2). 
    It can be seen that the ``orientation'' of the vertices (i.e., whether a vertex is a starting point or end point of the arrows) is preserved in the transition. 
    Notably, the orientation of the transition is {\it opposite} compared to that of up random walk. 
    (c) An example of harmonic random walk on $1$-simplices that moves to lower-adjacent simplices that are {\it not} upper-adjacent. 
    $K$ is some global normalization factor. 
    The laziness is proportional to the sum of the lower degree and the upper degree. }
    \label{fig:up_walk}
\end{figure}

Let us take the relationship among up random walk, up quantum walk, and up Laplacian for example. 
It should be emphasized that although the random walk is defined on the oriented simplicial complexes, the up Laplacian is defined for the simplices with a fixed orientation.  
In the up random walk, a $k$-simplex $\sigma$ transit to simplices $\sigma'$ that share a common upper simplex and the orientation of the $k+1$ induced by $\sigma$ and $\sigma'$ must be the same. 
In Figure~\ref{fig:up_walk}(a), it can be seen that the orientation of the triangle induced by the edges are same. 
This means that for $\sigma\neq \sigma'$,
$$
P^{up}_{\sigma,\sigma'}=\mathrm{Prob}(\sigma \rightarrow \sigma') 
= \frac{1}{K^{up}}
$$
if $\sigma$ and $\sigma'$ are up-adjacent (including the orientation) where $K^{up}$ is some normalization factor of probabilities. 
Meanwhile, the matrix element of the up Laplacian satisfies for $\sigma\neq \sigma'$ s.t. $\sigma$ and $\sigma'$ are up-adjacent, 
$$
(\Delta^{up}_k)_{\sigma_+,\sigma'_+}=
\begin{cases}
    1 & \text{if\ } \sigma_+\ \text{and}\ \sigma'_+\ \text{are\ up\ adjacent}, \\
    -1 & \text{if\ } \sigma_+\ \text{and}\ \sigma'_-\ \text{are\ up\ adjacent},
\end{cases}
$$
where $\sigma_+$ and $\sigma'_+$ (resp. $\sigma'_-$) are positively (resp. negatively) oriented simplices correspond to $\sigma$ and $\sigma'$, respectively. 
These relationships show that if a state transits to a state with the same (resp. different) orientation, the corresponding sign of the matrix element of the up Laplacian is $+$ (resp. $-$). 
Suppose we have a unitary encoding $U^{up}$ of $P^{up}$ using $\alpha$-ancilla qubits i.e., 
\begin{equation}
    \label{eq:Uup}\left(\Pi^\pm_{k}\otimes\bra{0^{\alpha}}\right)
    U^{up} \left(\Pi^\pm_{k}\otimes \ket{0^{\alpha}}\right)= \Pi^\pm_{k}P^{up}\Pi^\pm_{k},
\end{equation}
where $\Pi^\pm_k$ is the projector onto the space spanned by positively and negatively oriented $k$-simplices. 
{This is a unitary encoding of a Markov transition matrix $P^{up}$, which is implementable with the  classical circuit that simulates the transition $\ket{\sigma}\rightarrow P^{up}\ket{\sigma}$ ( See  Remark~\ref{remark:Pup} in Section~\ref{sec:def_up_random_walk}). 
Moreover, for the specific case of clique complexes, we show efficient construction of such unitary encoding in Theorem~\ref{theorem:clique} in Section~\ref{sec:clique_construction}. } 

{Then, we apply {$HZ$} gates to the qubit that represents the orientation to obtain a unitary $V^{up}:= \left((I_n\otimes HZ)\otimes I_{\alpha}\right)U^{up}$ using the $U^{up}$ in eq.~\eqref{eq:Uup}.}
Then, we apply {$HZ$} gates to the qubit that represents the orientation to obtain a unitary $V^{up}= \left((I_n\otimes HZ)\otimes I_{\alpha}\right)U^{up}$.
By this construction, we can add a relative phase factor of $-1$ if a simplex transit to other simplices with a different orientation. 
The laziness of the random walk (i.e., the probability of staying in the same simplex after a transition) is proportional to the up-degree (the number of $k+1$-simplices that a $k$-simplex has as upper simplex). 
The diagonal elements of the up Laplacian are given exactly by the up-degree of simplices. 
Therefore, we can conclude that $V^{up}$ is a unitary encoding of the up Laplacian. A formal statement and proof can be found in Proposition~\ref{prop:unitaryencoding_up} in Section~\ref{sec:main}. 
For the down and harmonic quantum walks, we can encode the down and full Laplacians with analogous arguments.

\paragraph{Implementing the unitary encoding of projectors.}

The implementation of the projectors is based on the quantum singular value transformations~\cite{gilyen2019quantum} of quantum walk unitaries. 
The $k$-chain space $C_k$, which is a space spanned by positively oriented $k$-simplices, can be decomposed as follows:
$$
C_k= B_k \oplus \mathcal{H}_k \oplus B^k,
$$
where $B_k$, $B^k$, $\mathcal{H}_k$ are $k$-boundary, $k$-coboundary, and $k$-harmonic homology group. 
For formal definitions, see Section~\ref{sec:preliminaries}.
These subspaces can be represented with Laplacian operators as follows:
\begin{itemize}
    \item $B_k=\mathrm{Im}(\Delta_k^{up})$, \ $B^k=\mathrm{Im}(\Delta_k^{down})$.
    \item $B_k\oplus \mathcal{H}_k = \ker({\Delta_k^{down}})$,  \ $B^k\oplus \mathcal{H}_k = \ker({\Delta_k^{up}})$.
\end{itemize}
Therefore, projectors onto these subspaces can be realized by the projectors onto the low/high energy subspace of $\Delta_k$, $\Delta_k^{up}$ and $\Delta_k^{down}$ with threshold $\lambda_k$, $\lambda_k^{up}$, and $\lambda_k^{down}$ where $\lambda_k$, $\lambda_k^{up}$, and $\lambda_k^{down}$ are the minimal non-zero eigenvalues (spectral gaps) of $\Delta_k$, $\Delta_k^{up}$, and $\Delta_k^{down}$, respectively. 
We perform quantum singular value transformations with polynomial approximation of rectangle functions with threshold $\lambda_k$, $\lambda_k^{up}$, and $\lambda_k^{down}$ with the quantum walk unitaries. 
The number of quantum walk unitaries used in the procedures is bounded by polynomial in $n$ if $\lambda_k$, $\lambda_k^{up}$, and $\lambda_k^{down}$ are lower bounded by $1/poly(n)$.

\paragraph{Construction of up/down quantum walk unitaries.}

We explain how to construct the quantum walk unitaries for clique complexes using the {membership function $$
f(\sigma)=\begin{cases}
    1 & \text{ if } \sigma \in X\\
    0 & \text{ if } \sigma \notin X.
\end{cases}
$$} 
For simplicity, we consider the unweighted case. 
The up quantum walk unitary is constructed with ``up-down'' construction and the down quantum walk unitary is constructed with ``down-up'' construction. 
In the ``up-down'' construction, $k$-simplices first go up to $k+1$-simplices and then go back to $k$-simplices. 
In the ``down-up'' construction, $k$-simplices first go down to $k-1$-simplices and then go back to $k$-simplices. 

When we go up from $k$-simplex $\sigma$ to $k+1$-simplices, we can change one of the 0's of the $n$-bit string $x_\sigma$ to 1.  
As the number of 0's in $x_\sigma$ is $n-k-1$, we use the superposition of $\sum_{i\in [n-k-1]} \frac{1}{\sqrt{n-k-1}}\ket{i}$ in an ancilla register. 
Of course, not all of the Hamming weight $k+2$ strings are included in the $k+1$-simplices. Therefore, we can apply the membership function to verify if they are included in the $k+1$-simplices. 
The induced orientation of the $k+1$-simplices depends on the position of $0$ that is converted to $1$. 
If there are odd numbers of 1 before the position of $0$, the relative orientation changes from the original simplex $\sigma$. 
If the number of the sum of 1 is even, the relative orientation does not change. 
After we transit to $k+1$-simplices, we similarly go down to $k$-simplices by changing one of the positions of $1$ to $0$ while tracking the orientation.
When we go down, we use the superposition of $\sum_{j\in [k+2]} \frac{1}{\sqrt{k+2}}\ket{j}$ prepared in another ancilla register. 

The ``down-up'' construction can be similarly done. Note that when we go down from $k$-simplices, the states must be included in the elements of $k-1$-simplices by the definition of simplicial complexes. Therefore, we apply the membership function after we go back to $k$-simplices. 
Note that in the ``up-down'' and ``down-up'' constructions, the self-loops (laziness) are automatically encoded. 

\paragraph{Construction of the harmonic walk unitary.}

For the construction of harmonic walk unitary, we use the following superpositions in an ancilla register:
$$
\sum_{i\in [n-k-1], j\in [k+1]} \ket{i,j}\ket{0} 
+
\sum_{l \in [n-k-1]} \ket{l,l}\ket{1}
+
\sum_{m \in [k+1]} \ket{m,m}\ket{1},
$$
where we hide a global normalization factor. 
The first term is for the adjacency, the second term is for the upper degree, and the third term is for the lower degree. 
Based on this superposition, we perform the following computation for each of the $\sigma$:
$$
\underbrace{\sum_{i\in [n-k-1], j\in [k+1]} \ket{i,j}\ket{0} }_{\substack{\text{Flip the\ } i\text{-th\ } 0\ \text{and\ } j\text{-th\ } 1\ \text{of}\ x_\sigma,\\ \text{then, check the membership.}}}
+
\underbrace{\sum_{l \in [n-k-1]} \ket{l,l}\ket{1}}_{\substack{\text{Check the membership of the simplex}\\ \text{obtained by flipping the}\ l\text{-th}\ 0\ \text{of}\ x_\sigma}}
+
\underbrace{\sum_{m \in [k+1]} \ket{m,m}\ket{1}.}_{\text{Do nothing.}}
$$
By Claim~\ref{claim:parity} in Section~\ref{sec:applications}, it follows that the relative orientation of the simplex obtained by the transition of the first term is equal to the sum of the bit elements of the positions between the indices of $i$-th 0 and $j$-th 1 of $x_\sigma$.
In the harmonic walk, simplices transit to other simplices that are lower-adjacent but not up-adjacent. This condition can also be computed using the membership function. 
Unlike the ``up-down'' and ``down-up'' construction, this construction separates the self-loops and adjacency. 

\paragraph{Applications.}

The constructed projectors onto topological subspace can be applied to problems of practical interest. 
The first two applications of our quantum walks are the tasks of estimating the normalized Betti numbers and the normalized persistent Betti numbers, that are related to the topological data analysis. 
First, the estimation of the normalized Betti numbers, denote by $\beta_k/n_kk$, can be done in the following steps:
\begin{itemize}
    \item Prepare the uniform mixture over $k$-simplices
    $
    \rho_k = \frac{1}{n_k} \sum_{\sigma\in X^+_k} \ket{\sigma}\bra{\sigma},
    $
    where $n_k$ is the number of $k$-simplices of a simplicial complex $X$.
    \item Estimate the energy w.r.t the combinatorial Laplacian $\Delta_k$ and measure the energy.
    \item Sample from the above quantum circuit and output: 
    $$
    \frac{\text{The}\ \#\ \text{of\ outcomes\ of\ energy\ below\ }\lambda_k }{\text{The}\ \#\ \text{of\ samples}}
    $$
    as the estimate, where $\lambda_k$ is the minimal non-zero eigenvalue of $\Delta_k$.
\end{itemize}
{We can estimate $\beta_k/n_k$ with the above procedures. This is because 
$$\rho_k=\frac{1}{n_k} \sum_{\sigma\in X^+_k} \ket{\sigma}\bra{\sigma}=\frac{1}{n_k} \sum_{i\in n_k} \ket{\psi_i}\bra{\psi_i}$$ where $\{\ket{\psi_i}\}_i$ are eigenvectors of $\Delta_k$ with eigenvalue $\lambda_i$ and the probability of observing $0$ energy is $\beta_k/n_k$.}
In the second step, we measure the energy w.r.t. the combinatorial Laplacian.
To be more precise, we actually construct a quantum circuit that can conditionally project onto the low-energy subspace of the combinatorial Laplacian. This means that in the third step, if the measurement outcome is $1$, the state is successfully projected onto the low-energy subspace below the energy threshold. 
The complexity of the procedure is governed by the complexity of preparing the initial mixed state, the complexity of energy measurement, and the number of necessary samples. 
We can efficiently estimate the normalized Betti numbers with inverse-polynomial additive precision
under the conditions that
we can efficiently sample from the uniform distribution over $k$-simplices and $\lambda_k> 1/poly(n)$. 

Next, we explain how to estimate the normalized persistent Betti numbers. 
We use the following relationship:
$
\beta^{i,j}_k = \dim (\ker (\Delta_k^{i,down}))- \dim (\ker (\Delta_k^{i,down}) \cap (\mathrm{Im}(\Delta_k^{j,up})),
$
where $i$ and $j$ represents the two simplicial complexes $X^i \subseteq X^j$. 
We can estimate $\beta^{i,j}_k/n_k$ by estimating
${\dim (\ker (\Delta_k^{i,down}))}/{n_k}\ \ \text{and}\ \ 
{\dim (\ker (\Delta_k^{i,down}) \cap (\mathrm{Im}(\Delta_k^{j,up}))}/{n_k},
$
respectively, 
where $n_k$ is the number of $k$-simplices of $X^i$. 
The first quantity can be estimated using the down quantum walk on $K_i$ and the second quantity can be estimated by combining the down quantum walk on $K_i$ and the up quantum walk on $K_j$ in a similar way with the case of Betti numbers. 
This application illustrates that down quantum walk and up quantum walks are useful for tasks related to persistence and TDA. 

Finally, in the case of the verification of the promise clique homology problem, the task is to verify if the homology group of the given clique complex is trivial or not. 
If the homology group is not trivial (YES case), we can receive a witness state from the prover which is supposed to be in the kernel of $\Delta_k$. 
In the NO case, there is a promise that $\Delta_k$ has some inverse-polynomial spectral gap. 
Therefore, we can show that the verifier can decide which is the case based on the implementation of the projector onto the low-energy subspace of $\Delta_k$ with desirable soundness and completeness. 

\subsection{Outlook and Related works}
\label{sec:related_works}

Our work opens avenues for applying quantum walks to higher-order networks. 
In this subsection, we discuss related works and open problems. 
{We first compare existing classical results and our results in detail.}
Second, we discuss the computational complexity of the problems related to clique complexes to clarify how hard the problems we are solving are and to highlight the evidence for the classical difficulty of the tasks. Next, we compare our results to another classical algorithm based on the path integral Monte Carlo method and existing QTDA algorithms. Finally, we present two open problems: one regarding the potential for further application of our construction to graph problems, and the other concerning the extension of our quantum walks to persistent homology groups.

\paragraph{Comparison with classical expectation process.}
We give a comparison between our results and previous results on classical random walks on simplicial complexes.
We give a summary of classical results in Table~\ref{tab:classical}. 
Although~\cite{schaub2020random} also provides conceptually similar results, we do not include their result because the rate of convergence of the result of~\cite{schaub2020random} is not known.  The formal definitions of up and down random walks can be found in Appendix~\ref{app:classical}. 
\begin{table}
    \centering
    \begin{tabular}{lcccc}
       Random walks   & NE Process $\tilde{\mathcal{E}}^{\sigma_0}_t$ & update  &$t\rightarrow \infty$  &  distance at $t$   \\ \hline
         \begin{tabular}{c} Up walk \\ \cite{parzanchevski2017simplicial} ~\eqref{NE:eq:up}\end{tabular} & $\left(\frac{k+1}{pk+1}\right)^t {\mathcal{E}}^{\sigma_0}_t$ & $I - \frac{1-p}{pk+1} \Delta^{up}_k$ & $Z^k$  & $\left(1-\frac{1-p}{pk+1}\lambda_k^{up}\right)^t$ \\ \hline
         \begin{tabular}{c}Down walk\\ \cite{mukherjee2016random} ~\eqref{NE:eq:down}\end{tabular}& $\left(\frac{M-1}{p(M-2)+1}\right)^t {\mathcal{E}}^{\sigma_0}_t$& $I - \frac{(1-p)\Delta_k^{down}}{(p(M-2)+1)(k+1)}$ & $Z_k$  & $\left(1 - \frac{(1-p)\lambda_k^{down}}{(p(M-2)+1)(k+1)}\right)^t$  \\ \hline
         \begin{tabular}{c}Harmonic walk\\ \cite{mukherjee2016random}   \end{tabular} & $\left(\frac{K}{p(K-1)+1}\right)^t {\mathcal{E}}^{\sigma_0}_t$& $I-\frac{(1-p)\Delta_k D^{-1}}{p(K-1)+1}$ & $\mathcal{H}_k$  & $\left(1-\frac{1-p}{p(K-1)+1}\lambda_k\right)^t$ \\\hline
    \end{tabular}
    \caption{
    Classical results on simplicial complexes and homology. 
    NE process referees to normalized expectation process. The expectation process is ${\mathcal{E}}^{\sigma_0}_t = p^{\sigma_0}_t(\sigma) - p^{\sigma_0}_t(\overline{\sigma})$, where $p^{\sigma_0}_t(\sigma)$ is a probability of being in $\sigma$ after $t$-steps of random walks starting from $\sigma_0$.
    The ``update'' means operators $A$ such that  $\tilde{\mathcal{E}}^{\sigma_0}_{t+1}=A\tilde{\mathcal{E}}^{\sigma_0}_t$ for each of the random walks. 
    The column of ``$t\rightarrow \infty$'' shows the final subspaces to that every initial state $\sigma_0$ will reach with the normalized expectation processes. 
    The ``distance'' means orders of distance between $\tilde{\mathcal{E}}^{\sigma_0}_\infty$ and $\tilde{\mathcal{E}}^{\sigma_0}_t$ in the operator norm. 
    The number of time steps $t$ should be comparable to the number of quantum walk unitaries in our results. 
    $M$ is the maximal (up) degree of $k-1$ simplices. 
    $D^{-1}$ is a diagonal matrix such that $(D^{-1})_{\sigma,\sigma}=((\Delta_k)_{\sigma,\sigma})^{-1}$ if $(\Delta_{k})_{\sigma,\sigma'}\neq 0$ and $(D^{-1})_{\sigma,\sigma}=0$ otherwise. $K = \mathcal{O}(\|\Delta_k D^{-1}\|)$. 
     }
    \label{tab:classical}
\end{table}
The number of required time steps $t$ should be comparable to the number of uses of quantum walk unitaries in our results. 
As it can be seen from the table, the rate of the convergence to the target subspaces is similar to our results in the sense that the distance can be made to be inverse polynomially small in polynomial time steps $t$ with proper choice of the laziness if the spectral gaps of the Laplacians are at least inverse-polynomially large. 
However, there are two obstacles to {\it classically} simulating the normalized expectation process. 
One obstacle is that in order to perform the process, one needs to characterize the behavior of two different probabilities at the same time to take their differences. Moreover, the difference of two probabilities must be multiplied by an exponentially large normalization factor. 
These are highly challenging obstacles to classically simulate the normalized expectation process. 

\paragraph{Classical random walk without the orientation problem}
{
Although the problem of orientation appears as an obstacle for classical random walks on simplicial complexes, there are cases where such an obstacle can be overcome. For example, if the underlying simplicial complex is a triangulation of an orientable manifold, there are classical random walks that do not suffer from the orientation problem~\cite{eidi2023irreducibility}. 
Furthermore, it has been recently shown that when the underlying simplicial complex satisfies certain conditions regarding the orientation, there is an MA-complete problem related to the homology of simplicial complexes~\cite{hayakawa2025computational}. The MA-completeness indicates that the problem can be verified with a classical random walk without suffering from the orientation problem. These results indicate that there are subclasses of problems related to high-dimensional topology that are feasible for classical random walks, although they are not efficient for general simplicial complexes.}

{
\paragraph{Outlook on quantum advantage.}

Our prospects of superpolynomial speedup based on quantum walks come from the application for the problems related to the topological data analysis and the homology of clique complexes. 
Here, we discuss the evidence for classical hardness, best-known classical algorithms, and regimes for potential superpolynomial speedup. 

\begin{itemize}
\item {\bf Computational complexity of estimating the normalized Betti numbers:} 
We have shown an application of estimating the normalized Betti numbers (Theorem~\ref{thm:Betti_numbers}). 
First, the normalized Betti number estimation (NBE) problem for inverse-polynomial precision for {\it general} simplicial complexes is \DQC-hard~\cite{cade2021complexity}. 
Similarly, the Low-lying Energy Density Problem, a generalization of the NBE problem, is also \DQC -hard~\cite{gyurik2022towards}. 
These results can be seen as evidence that there are no efficient classical algorithms for such generalized problems assuming that \DQC -hard problem cannot be solved efficiently classically. 
Note that it is still open whether the estimation problem of normalized Betti numbers in inverse polynomial precision is \DQC-hard for clique complexes for which we have efficient quantum algorithms. 
\item {\bf Known classical algorithms for estimating the normalized Betti numbers:}
As we have discussed, a general classical algorithm for NBE problem in inverse-polynomial additive error would lead to a classical algorithm for \DQC-hard problems, which is very surprising. 
It has been shown that this is not the case for the estimation in the constant additive error. 
In~\cite{apers2023simple}, the authors proposed an efficient path integral Monte Carlo algorithm NBE problem in a constant additive error. 
The algorithm of~\cite{apers2023simple} first samples uniformly from the $k$-th simplices and then performs some random walk related to the combinatorial Laplacian.
Finally, it computes and outputs the estimate for the normalized Betti numbers. 
Instead of sampling from the space with positive and negative orientations, they perform random walks on unoriented $k$-simplices and use the information of orientation in the post-processing. 
\item {\bf Regimes for superpolynomial speedup:} 
For the application in NBE problem, three conditions should be satisfied so that the quantum algorithm works efficiently. First, we should be able to sample from the uniform distribution over $k$-simplices. Second, the combinatorial Laplacian should have an inverse-polynomial spectral gap. Third, the $k$-Betti number should be exponentially large (because the normalized quantity may be exponentially large and we can only estimate the normalized Betti number in the inverse-polynomial additive error.) 
In~\cite{berry2024analyzing}, it is shown that there exists a family of simplicial complexes that satisfy all these criteria and therefore can be a regime of simplicial complexes for superpolynomial quantum speedup. 
At the same time, it is also shown that for random clique complexes (clique complexes for Erdős–Rényi), only approximately a quartic speedup can be achieved~\cite{berry2024analyzing}. 

\item {\bf Estimating the (non-normalized) Betti numbers:} 
We have seen there can be superpolynomial speedup for the estimation of normalized (persistent) Betti numbers. 
What can be said for the estimation of the (non-normalized) Betti numbers, which is more practically motivated? 
For this problem, it has been discussed in~\cite{mcardle2022streamlined} that there is at most polynomial quantum speedup with the currently known quantum algorithms over the known classical algorithms~\cite{mischaikow2013morse,milosavljevic2011zigzag}. 
Therefore, it is still open if a quantum algorithm can achieve any practical task related to TDA with superpolynomial speedup.
\end{itemize}
}

\paragraph{Comparison between the existing QTDA algorithms.}
We give a comparison between our algorithm and the prior works on QTDA~\cite{lloyd2016quantum, hayakawa2022quantum, gilyen2019quantum}. 
Compared to the prior works, we {\it double} the state space due to the distinction of orientations. Note that, however, this doubling only requires an additional single qubit. 
It can be said that it is due to this addition of a single qubit that we have been able to ease the ``sign problem'' of combinatorial Laplacians\footnote{ Here, the ``sign problem'' referees to the fact that there are both positive and negative entries in the off-diagonal elements of the combinatorial Laplacian.} and relate QTDA to Markov chains. 
Therefore, although it is not that our algorithm drastically improves the time complexity for QTDA compared to the prior works on QTDA, we have been able to find a new way to perform topological information processing through the experimental realization of quantum walks \cite{tang2018experimental, acasiete2020implementation, dadras2019experimental, qiang2024review}, and in a broader sense, our work opens a way to unify the study of quantum walks and QTDA. 

\paragraph{Lackadaisical quantum walks.}

Our random walks on simplicial complexes are {\it lazy} random walks i.e., there are self-loops. 
A lazy version of quantum walks is called the lackadaisical quantum walks~\cite{wong2015grover}~\cite{hoyer2020analysis}. 
Therefore, our quantum walks can be seen as high-dimensional lackadaisical quantum walks. 
If we consider our quantum walks on $0$-simplices i.e., vertices, we can straightforwardly obtain the encoding of graph Laplacians into quantum walk unitaries. 
As graph Laplacian is an important tool in graph-based information processing, can the encoding of graph Laplacians into quantum walks find more applications? 

\paragraph{Persistent harmonic homology group.}

In this paper, we show the implementation of several projectors onto subspaces related to topology. 
Recently, persistent Laplacian operators are introduced~\cite{memoli2022persistent}~\cite{wang2020persistent}~\cite{wei2023persistent}. 
It has been known that the dimension of the kernel of the persistent Laplacian is equivalent to the persistent Betti numbers. 
Although we show how to estimate the normalized persistent Betti numbers in Section~\ref{sec:app:persistent_Betti_numbers}, it is based on the combination of several projectors and not based on the persistent Laplacian. 
Is it possible to construct quantum walks with which we can directly implement the projector onto the {\it persistent harmonic homology group} (i.e., the kernel of the persistent Laplacian)?

{
\paragraph{High-dimensional expanders.}
Random walks on simplicial complexes are actively studied in the context of high-dimensional expanders~\cite{kaufman2020high}. Although they do not typically consider the orientation and therefore the motivation is different, it would be interesting to explore potential connections between the study of high-dimensional expanders and quantum walks on simplicial complexes. 
}

\paragraph{Quantum walk algorithm without QSVT.} 
Our quantum algorithms first construct quantum walk unitaries that encode a single step of random walks on simplicial complexes. Then we utilize the quantum singular value transformation. This is because if we try to quantize a {\it sequence} of random walk steps rather than a single step, we would also suffer from the problem of an exponentially large normalization factor, as is the case in classically simulating the normalized expectation process. 
We leave it for future work to construct quantum walk algorithms on simplicial complexes that avoid the problem of exponentially large normalization without QSVT.  

\paragraph{Subsequent works.}
{
Since the appearance of this work, quantum algorithms for higher-order 
networks and TDA have been further developed. 
Leditto et al.~\cite{leditto2025quantum} 
proposed Quantum HodgeRank as an application of quantum topological 
signal processing on simplicial complexes~\cite{leditto2025topological}. 
Gyurik et al.~\cite{gyurik2024quantum} introduced the harmonic 
persistence problem and established its computational complexity. 
Extending our quantum walk framework to these problems is an 
interesting direction for future work. 
More broadly, finding further applications of quantum walks on 
higher-order networks, as we have done for the HDDP in this work, 
remains an important open problem.
}

\subsection{Organization of the paper}

The remainder of this paper is organized as follows. In Section
~\ref{sec:preliminaries}, we introduce several preliminaries on simplicial complexes, homology, and quantum computing. 
In Section~\ref{sec:main}, we present the implementation of the projectors based on the quantum walks. This section includes the formal definitions of up, down, and harmonic random walks on simplicial complexes.  
In Section~\ref{sec:clique_construction}, we show efficient constructions of up, down, and harmonic quantum walk unitaries for clique complexes.
In Section~\ref{sec:applications}, we show two applications of our quantum walks on simplicial complexes: estimation of the normalized Betti numbers, estimation of the normalized persistent Betti numbers.
{Finally, in Section~\ref{sec:dirichlet}, we show application of our quantum walk framework on HDDP.}

\section{Preliminary}
\label{sec:preliminaries}
In this section, we introduce several preliminaries. 
In Section~\ref{sec:pre:simplicialcomplexes}, we introduce simplicial complexes and their orientations. 
In Section~\ref{sec:pre:adjacency}, we introduce up-adjacency and down-adjacency of simplicial complexes. 
In Section~\ref{sec:pre:homology} we introduce boundary operators and homology group. 
In Section~\ref{sec:pre:hodge}, we introduce the discrete Hodge theory, which relates the homology and the Laplacian operators. 
In Section~\ref{sec:pre:matrix_representation} and ~\ref{sec:pre:weighting}, we introduce the matrix representation of Laplacians and the weighting of simplicial complexes. 
In Section~\ref{sec:pre:quantum}, we introduce some preliminaries on quantum computing. 
Finally, in Sectioin~\ref{sec:pre:notations}, we further introduce some notations. 

\subsection{Oriented simplicial complexes} 
\label{sec:pre:simplicialcomplexes}
An abstract simplicial complex $X$ over a finite set of vertices $V\subset \nats$ is a set of vertices of $V$ that is closed under taking subsets. 
Let $n\coloneqq |V|$. A simplex is called a $k$-simplex if it is composed of $(k+1)$-vertices. 
We denote the set of $k$-simplices of $X$ as $X_k$.
We denote the number of $k$-simplices of $X$ by $n_k$.
An oriented $k$-simplex $\sigma$ is a simplex $\sigma$ with orientation inherited from the ordering of vertices $\sigma=[v_0,v_1,...,v_k]$. Two oriented $k$-simplices $\sigma=[v_0,v_1,...,v_k]$ and $\sigma'=[v'_0,v'_1,...,v'_k]$ are said to have the same orientation if $[v_{\pi(0)},v_{\pi(1)},...,v_{\pi(k)}]=[v'_0,v'_1,...,v'_k]$ for an even permutation $\pi$. 
Such $\sigma_k$ and $\sigma_k'$ are regarded as the same oriented simplex. 

Throughout the paper, we say $\sigma_k=[v_0,v_1,...,v_k]$ is positively (negatively) oriented if an ordering $[v_{\pi(0)},v_{\pi(1)},...,v_{\pi(k)}]$ such that $v_{\pi(0)}<v_{\pi(1)}<\cdots<v_{\pi(k)}$ is obtained with an even (odd) permutation $\pi$. 
We denote the set of positively oriented $k$-simplices by $X^+_k$ and negatively oriented $k$-simplices by $X^-_k$. It follows that $|X^+_k|=|X^-_k|=n_k$.
For any oriented simplex $\sigma\in X^\pm_k $, $\overline{\sigma}$ represents its counterpart with the opposite orientation.
For any $\sigma\in X^\pm_k $, we can identify its corresponding unoriented simplex with a Hamming weight $k+1$ bit string $x_\sigma\in\{0,1\}^n$, where the positions of $1$s correspond to the membership of the vertices of the simplex. 
With a slight abuse of notation, 
we denote the set of all Hamming weight $k+1$ bit strings  associated with elements of $X^\pm_k$ by $X_k\subseteq \{0,1\}^n$. 

\subsection{Adjacency of simplicial complexes}
\label{sec:pre:adjacency}
For an unoriented simplex $\sigma=\{v_0,v_1,...,v_k\}$, its faces are the subsets obtained by removing a single vertex: $\sigma\backslash v_0,\sigma\backslash v_1,...,\sigma\backslash v_k$.  By definition,  any $k$-simplex has exactly $k+1$ faces.   
An oriented simplex $\sigma=[v_0,v_1,...,v_k]$ induces an orientation on its faces as $(-1)^j[v_0,...,v_{j-1},v_{j+1},...,v_k]$
where $(-1)^j$ means to take the opposite orientation if $(-1)^j=-1$. 
(We take the opposite orientation if we remove the $j$-th vertex and $j$ is odd.)
A $(k+1)$-simplex $\sigma'_{k+1}$ is called a coface of $\sigma_k$ if $\sigma_k$ is a face of $\sigma'_{k+1}$.
For $\sigma, \sigma' \in X^\pm_k$, we use the following notations:
\begin{itemize}
        \item  (Up-adjacency:) We denote $\sigma \sim_{\uparrow} \sigma'$ if they share a common $k+1$-dimensional coface and the orientation of the coface induced by $\sigma$ and $\sigma'$ are same. If $\sigma \sim_{\uparrow} \sigma'$, $\sigma$ and $\sigma'$ are said to be oriented similarly. 
        If $\sigma \sim_{\uparrow} \overline{\sigma'}$, $\sigma$ and $\sigma'$ are said to be oriented dissimilarly. 
        \item  (Down-adjacency:) We denote $\sigma\sim_{\downarrow}\sigma'$ if they share a common $k-1$-dimensional face and the orientation of the face induced by $\sigma$ and $\sigma'$ are same. 
        If $\sigma \sim_{\downarrow} \sigma'$, $\sigma$ and $\sigma'$ are said to have a similar common lower simplex. 
         If $\sigma \sim_{\downarrow} \overline{\sigma'}$, $\sigma$ and $\sigma'$ are said to have a dissimilar common lower simplex.
        \item $\sigma\nsim_{\uparrow}\sigma'$: if $\sigma$ and $\sigma'$ are not up-adjacent. 
        \item $\deg(\sigma)$: number of cofaces of $\sigma$. 
\end{itemize}
The quantity $\deg(\sigma)$ should also be called as {\it up}-degree as it is the number of cofaces. 
We simply refer to it as ``degree" because the number of faces of any $k$-simplex (down-degree) is always $k+1$.
Note that for any $k$-simplices $\sigma$ and $\sigma'$, the common face or coface is always unique if they exist.

\subsection{Chain group and homology}
\label{sec:pre:homology}
We denote a $k$-th chain group $C_k$ defined as $\mathrm{Span}\{\ket{\sigma}\}_{\sigma\in X_k^+}$ with real coefficients. 
We can define a $k$-th boundary operator $\partial_k:C_k\rightarrow C_{k-1}$ as an operator that acts for any $k$-simplex in $[v_0,v_1,...,v_{k}]\in X^+_k$ as
$$
    \partial_k [v_0,v_1,...,v_{k}]
    =
    \sum_{j=0}^k (-1)^j [v_0,...,v_{j-1},v_{j+1},...,v_k].
$$
The important property of boundary operators is that $\partial_{k-1}\partial_k=0$. 
The $k$-th homology group is defined as $$H_k\coloneqq \ker(\partial_{k})/\mathrm{Im}(\partial_{k+1}), $$ and the $k$-th Betti number is defined as $\beta_k\coloneqq \dim(H_k)$. 
$B_k\coloneqq \mathrm{Im}(\partial_{k+1})$ is called the $k$-boundary and $Z_k\coloneqq \ker(\partial_k)$ is called the $k$-cycle. 

Given two simplicial complexes $X^i \subseteq X^j$ over the same vertices, we can define the persistent homology groups and the persistent Betti numbers. 
The $k$-th persistent homology group is defined as follows: 
$$
H^{i,j}_k \coloneqq \ker(\partial_k^i)/(\ker(\partial_k^i)\cap \mathrm{Im}(\partial_{k+1}^j)).
$$
The $k$-th persistent Betti number is defined as $\beta^{i,j}_k\coloneqq \dim(H^{i,j}_k)$.

\subsection{Discrete Hodge theory}
\label{sec:pre:hodge}
We define the inner product such that $\{\ket{\sigma}\}_{\sigma\in X^+_k}$ forms an orthonormal basis for $C_k$. This means that $\braket{\sigma'|\sigma}= \delta_{\sigma,\sigma'}$ for all $\sigma,\sigma'\in X^+_k$. 
Then, let $\partial^*_k$ be an adjoint of $\partial_k$. 
With our choice of inner product, $\partial^*_k=\partial^\dagger_k$. 
The $k$-th combinatorial Laplacian (Hodge Laplacian) $\Delta_k:C_k\rightarrow C_k$ is defined as
$$\Delta_k\coloneqq \partial_{k+1}\partial_{k+1}^*
+\partial_k^*\partial_k.$$
We also define
$$\Delta_k^{up}\coloneqq \partial_{k+1}\partial_{k+1}^*,$$
$$\Delta_k^{down}\coloneqq \partial_k^*\partial_k.$$
It is known that 
    $$
    \ker(\Delta_k)\cong H_k
    $$
and 
    $$
    \beta_k=\mathrm{Nullity}(\Delta_k)\coloneqq \dim(\ker(\Delta_k)).$$
We denote the smallest non-zero eigenvalues of $\Delta_k$, $\Delta_k^{up}$ and $\Delta_k^{down}$ as $\lambda_k$, $\lambda_k^{up}$ and $\lambda_k^{down}$, respectively. 
{
It holds that $\lambda_k= \min \{\lambda_k^{up}, \lambda_k^{down}\}$ because the non-zero energy spectrum of $\Delta_k$ is the union of the non-zero spectrums of $\Delta_k^{up}$ and $\Delta_k^{down}$. 
}
Note that the spectrum of the Laplacians does not depend on the choice of the orientation of simplices.

We can decompose $C_k$ as 
$$
C_k = \mathrm{Im}(\partial_{k+1})\oplus \ker(\Delta_k) \oplus \mathrm{Im}(\partial_k^*).
$$
This decomposition is known as the Hodge decomposition. 
It holds that $\mathrm{Im}(\partial_{k+1})=\mathrm{Im}(\Delta_k^{up})$ and $\mathrm{Im}(\partial_k^*)=\mathrm{Im}(\Delta_k^{down})$. 
Therefore,
we can decompose $C_{k}$ in terms of up, down Laplacians as  
$\mathrm{Im}(\Delta_k^{up})\oplus
\ker(\Delta_k)\oplus
\mathrm{Im}(\Delta_k^{down})$ 
. 

\subsection{Matrix representation of the unweighted combinatorial Laplacian}
\label{sec:pre:matrix_representation}

The combinatorial Laplacian can be described using the degree and the adjacency relationships. 
It is known that 
$\Delta^{up}$, $\Delta^{down}$ and $\Delta_k$ satisfy for all $\sigma,\sigma' \in X^+_k$,
    \begin{equation}
    (\Delta^{up}_k)_{\sigma,\sigma'}= \begin{cases}
        \deg(\sigma), & if\ \sigma=\sigma' \\
        1, & if\ \sigma\sim_{\uparrow}{\sigma'}\\
        -1, & if\ \sigma\sim_{\uparrow}\overline{\sigma'}\\
        0, & otherwise.
    \end{cases}
    \end{equation}
    \begin{equation}
    (\Delta^{down}_k)_{\sigma,\sigma'}=
    \begin{cases}
        k+1, & if\ \sigma=\sigma' \\
        1, & if\ \sigma\sim_{\downarrow}{\sigma'}\\
        -1, & if\ \sigma\sim_{\downarrow}\overline{\sigma'}\\
        0, & otherwise.
    \end{cases}
    \end{equation}
    \begin{equation}
    (\Delta_k)_{\sigma,\sigma'}=
    \begin{cases}
        \deg(\sigma)+k+1, & if\ \sigma=\sigma' \\
        1, &if\ \sigma\sim_{\downarrow}{\sigma'}\ and\  \sigma\nsim_{\uparrow}\overline{\sigma'}
        \\
        -1, &if\ \sigma\sim_{\downarrow}\overline{\sigma'}\ and\ \sigma\nsim_{\uparrow}\sigma'\\
        0, & otherwise.\ 
    \end{cases}
    \end{equation}
Note that $\sigma\sim_{\uparrow}\sigma'$ also implies $\sigma\sim_{\downarrow}\overline{\sigma'}$. 
Therefore, in the full Laplacian, the matrix elements of $\sigma, \sigma'$ that are only lower adjacent (regardless of orientations) remain to be non-zero.

\subsection{Vertex-weighted simplicial complexes}
\label{sec:pre:weighting}

We generalize the simplicial complexes with a vertex-weighted graph $G$, which is characterized by a function {$w:V\rightarrow (0,\infty]$} assigning each vertex $v$ a weight $w(v)$ following~\cite{king2023promise}.
A vertex-weighted simplicial complex $X$ over the vertices is defined in the same manner as the unweighted case. 
The main difference between them is that we assign 
each simplex $\sigma$ in $X$ a weight $w(\sigma)$, defined as the product over weights of all vertices in $\sigma$: 
$$
w(\sigma)=\prod_{v\in \sigma}w(v).
$$
Furthermore, the inner product on the $k^{\text{th}}$ vertex-weighted simplicial complexes chain space, $C_{k}(X,\mathbb{R})$, is defined by 
$$
\langle\sigma \mid \tau\rangle= \begin{cases}w(\sigma)^2 & if\ \sigma=\tau ,\\ 0 & { otherwise, }\end{cases}
$$
for $\sigma,\tau \in X_{k}$ are $k$-simplices of $X$. 
This inner product implies that the set $\{ \ket{\sigma'}:=\frac{1}{w(\sigma)}\ket{\sigma}: \sigma \in C_{k} \}$ forms an orthonormal basis for $C_{k}$.

To define the combinatorial Laplacian $\Delta_{k}$ of the $k^{\text{th}}$ vertex-weighted simplicial complexes chain, it is necessary to define the boundary operator $\partial_{k}$: 
$$
\partial_k\left|\sigma^{\prime}\right\rangle=\sum_{v \in \sigma} w(v)\ket{(\sigma \backslash\{v\})^{\prime}}.
$$

With the boundary operator, the combinatorial Laplacian $\Delta_{k}=\partial_{k+1}\partial_{k+1}^*+\partial_k^*\partial_k$ can be expressed as follows: $\Delta^{up}$, $\Delta^{down}$ and $\Delta_k$ satisfy for all $\sigma,\sigma' \in X^+_k$,
    \begin{equation}
    \label{eq:up_laplacian}
    (\Delta^{up}_k)_{\sigma,\sigma'}= \begin{cases}
        \sum_{u \in up(\sigma)} w(u)^2, & if\ \sigma=\sigma' \\
        w(v_{\sigma})w(v_{\sigma'}), & if\ \sigma\sim_{\uparrow}{\sigma'}\\
        -w(v_{\sigma})w(v_{\sigma'}), & if\ \sigma\sim_{\uparrow}\overline{\sigma'}\\
        0, & otherwise.
    \end{cases}
    \end{equation}
    \begin{equation}
    \label{eq:down_laplacian}
    (\Delta^{down}_k)_{\sigma,\sigma'}=
    \begin{cases}
        \sum_{v \in \sigma} w(v)^2, & if\ \sigma=\sigma' \\
        w(v_{\sigma})w(v_{\sigma'}), & if\ \sigma\sim_{\downarrow}{\sigma'}\\
        -w(v_{\sigma})w(v_{\sigma'}), & if\ \sigma\sim_{\downarrow}\overline{\sigma'}\\
        0, & otherwise.
    \end{cases}
    \end{equation}
    \begin{equation}
    \label{eq:laplacian}
    (\Delta_k)_{\sigma,\sigma'}=
    \begin{cases}
        \left(\sum_{u \in up(\sigma)} w(u)^2\right)+\left(\sum_{v \in \sigma} w(v)^2\right) , & if\ \sigma=\sigma' \\
        w(v_{\sigma})w(v_{\sigma'}), &if\ \sigma\sim_{\downarrow}{\sigma'}\ and\  \sigma\nsim_{\uparrow}\overline{\sigma'}
        \\
        -w(v_{\sigma})w(v_{\sigma'}), &if\ \sigma\sim_{\downarrow}\overline{\sigma'}\ and\ \sigma\nsim_{\uparrow}\sigma'\\
        0, & otherwise.\ 
    \end{cases}
    \end{equation}
Here, { $\operatorname{up}(\sigma) \subset V$} is the subset of vertices $v$ such that $\sigma \cup\{v\} \in X_{k+1}$ is a $k+1$-simplex, and $v_{\sigma}$ and $v_{\sigma'}$ are the vertices removed from $\sigma$ and $\sigma'$, respectively, to obtain the common lower simplex.

\subsection{Unitary encoding, transformation, and measurement}
\label{sec:pre:quantum}

We introduce some preliminaries on quantum computing.
The block-encoding of an operator is defined as follows:
\begin{definition}[\cite{gilyen2019quantum}]
\label{def:unitary_encoding}
    Suppose that $A$ is an $s$-qubit operator, $a,\epsilon \in \mathbb{R}_+$, and $\alpha\in \mathbb{N}$. We say that the $(s+\alpha)$-qubit unitary $U$ is an $(a,\alpha,\epsilon)$-block-encoding of $A$, if 
    $$
    \|
    A-a(I_s\otimes \bra{0^\alpha}) U (I_s\otimes \ket{0^\alpha} )
    \|
    \leq \epsilon.
    $$
    More generally, we say that $U$ is a {$\Pi_s$-projected $(a,\alpha,\epsilon)$-unitary encoding of $A$} if 
    $$
    \|
    A-a (\Pi_s \otimes \bra{0^\alpha}) U (\Pi_s \otimes \ket{0^\alpha} )
    \|
    \leq \epsilon,
    $$
where $\Pi_s$ is an $s$-qubit projector.
\end{definition}
Using the block encoding or the unitary encoding, we can perform quantum singular value transformation (QSVT)~\cite{gilyen2019quantum}. 
For a Hamiltonian $H$ with eigendecomposition $H=\sum_i \lambda_i \ket{\psi_i} \bra{\psi_i}$, the singular value transformation of $H$ with polynomial $f(x)$ is 
$$
f(H)=\sum_i f(\lambda_i) \ket{\psi_i} \bra{\psi_i}.
$$

We introduce a $\mathrm{C}_\Pi \mathrm{NOT}$-gate which applies $X$ gate conditioned that a state is in the image of a projector $\Pi$ as:
$$
\mathrm{C}_\Pi \mathrm{NOT}\coloneqq X\otimes \Pi + I\otimes (I-\Pi).
$$

\paragraph{Block-measurement.}
Suppose we have $U_\Pi$ that is a block encoding of a projector $\Pi$, i.e., 
$$
\|
\Pi- (I_s \otimes \bra{0^a}) U_\Pi (I_s \otimes \ket{0^a} )
\|
\leq \epsilon. 
$$
Then, we can perform a measurement procedure called ``block-measurement''~\cite{rall2021faster}. 
The circuit for block-measurement is defined as 
$$
U_{\mathcal{M}}= 
 (U_\Pi^\dagger\otimes I_\alpha)\cdot \mathrm{C_{\Pi_a}NOT} \cdot (U_\Pi \otimes I_\alpha),
$$
where $\alpha$ denotes the single ancilla qubit and $\Pi_a = I_s \otimes \ket{0^a}\bra{0^a}$.
Let $\tilde{\Lambda}_\Pi$ be a quantum channel by tracing out the $a$ qubits of $U_{\mathcal{M}}$. Then, the following lemma is shown:
\begin{lemma}[\cite{rall2021faster}]
\label{lemma:block_measurement}
    Let $\Lambda_\Pi$ be a quantum channel that maps as 
    $$
    \ket{0}_\alpha\otimes \ket{\psi} \rightarrow
    \ket{1}_\alpha \otimes \Pi \ket{\psi} + \ket{0}_\alpha  \otimes (I_s-\Pi)\ket{\psi}.
    $$
    Then 
    $$
    \| \Lambda_\Pi - \tilde{\Lambda}_\Pi \|_\diamond = O(\epsilon),
    $$
    where $\|\cdot\|_\diamond$ is the diamond norm. 
\end{lemma}
We call the ancilla qubit labeled by $\alpha$ as the {\it output qubit} of the block-measurement.

\subsection{Notation}
\label{sec:pre:notations}

In this subsection, we introduce some notation. 
We use a special state $\Theta$, which is called an absorbing state in the literature of random walks. 
Let $S_k\coloneqq X^\pm_k\cup \{\Theta\}$. 
This is the state space of our random walks on simplicial complexes. 
{ Note, however, that we do not expect the introduction of the absorbing state is essential one to construct quantum walks on simplicial complexes. It is introduced to relate our quantum walk unitaries to Laplacian operators.
}
For any $\ket{\sigma}$ where $\sigma\in S_k$,
we identify an $(n+2)$-qubit computational basis state as follows: 
\begin{equation}
\ket{\sigma}=
\begin{cases}
    \ket{x_{\sigma}}\otimes \ket{00} & if\ \sigma\in X^+_k, \\
    \ket{x_\sigma}\otimes \ket{10} & if\ \sigma\in X^-_k, \\
    \ket{0^n} \otimes \ket{01} & if\ \sigma=\Theta, \\
\end{cases}
\end{equation}
where $x_\sigma\in X_k$ is the Hamming weight $k+1$ $n$-bit string that represents the unoriented simplex that corresponds to $\sigma$.
We also use 
$$
s(\sigma)= 
\begin{cases}
    0 & if\ \sigma \in X^+, \\
    1 & if\ \sigma \in X^-. 
\end{cases}
$$
Therefore, for $\sigma\in X^\pm_k$, 
\begin{equation}
\ket{\sigma}=
    \ket{x_{\sigma},s(\sigma),0}. 
\end{equation}
For $\sigma\in X^\pm_k$, we use $\sigma_+$ as
$\sigma_+=\sigma$ if $\sigma\in X^+_k$ and $\sigma_+=\bar{\sigma}$ if $\sigma\in X^-_k$. 
Then, $\sigma_-\coloneqq \overline{\sigma_+}$. 
{Also, let $\Pi_k$ be a projector onto the space spanned by $\{\ket{\sigma}\}_{\sigma\in X^+_k}=\{\ket{x_\sigma}\otimes \ket{00}\}_{\sigma\in X^+_k}$
and let $\Pi^\pm_k$ be a projector onto the space spanned by $\{\ket{\sigma}\}_{\sigma\in X^\pm_k}$
.}

For $\sigma\in X^\pm_k$, we also introduce the following notation:

\begin{itemize}
\item $i(\sigma)\in [n]$: the position of $i$-th 0 of $x_{\sigma}$.  ($i\in [n-k-1]$.)
\item $\tilde{i}(\sigma) \in [n]$: the position of $i$-th 1 of $x_{\sigma}$. ($i\in [k+1]$.)
\end{itemize}

\section{Quantum walks and homology of simplicial complexes}
\label{sec:main}
We begin this section by introducing our main results on the implementation of projectors which contains three folds in Section~\ref{sec::Formal statements of main results}. 
First, we propose a Markov transition matrix $P^{up}$, which could aid in encoding two projectors. One projector has support consisting of the $k$-cocycle, while the other has support consisting of the $k$-boundary. 
Second, we employ a similar approach, substituting $P^{up}$ with another transition matrix $P^{down}$, to encode two projectors. One has support consisting of the $k$-cycle, while the other has support consisting of the $k$-coboundary.
Third, one projector of the $k$-th homology group is encoded using another  transition matrix $P$.
The relationships among the subspaces are depicted in Figure~\ref{fig:hodgedecomposition}. 
\begin{figure}
    \centering
    \includegraphics[width=0.65\linewidth]{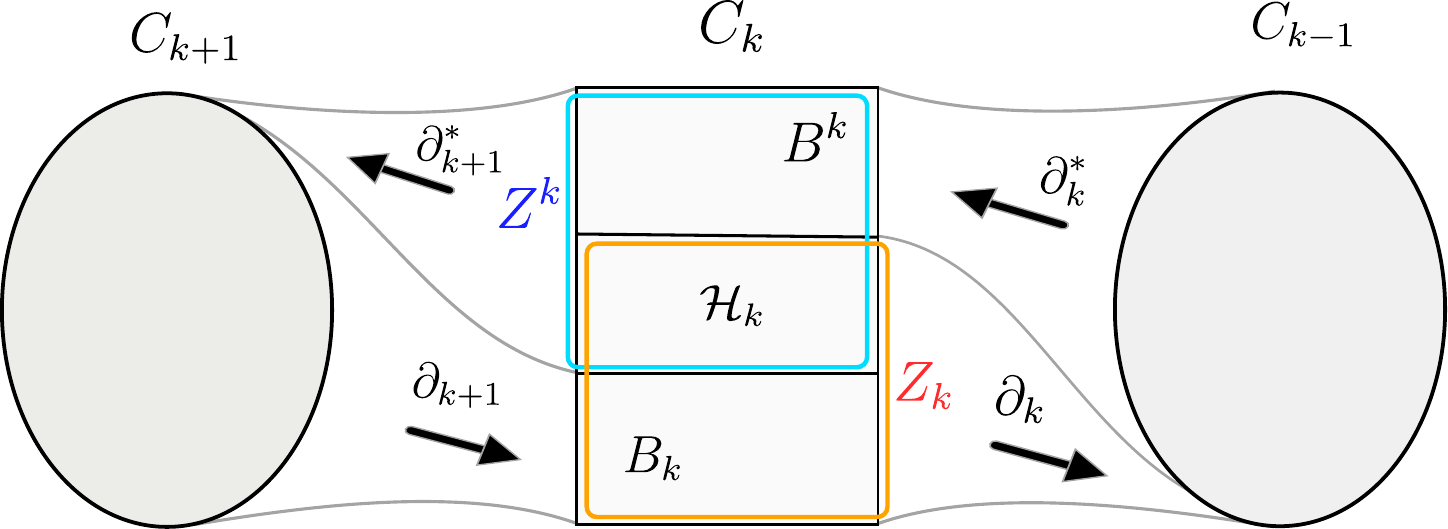}
    \caption{Relationships among the $k$-(co)cycle $Z_k$ ($Z^k$), $k$-(co)boundary $B_k$ ($B^k$), and $k$-harmonics $\mathcal{H}_k$.}
    \label{fig:hodgedecomposition}
\end{figure}

Then, the proofs for these results are presented in Section~\ref{sec::Proofs outline}, which consists of three steps. 
First, we formally define the Markov transition matrices $P^{up}$, $P^{down}$, and $P$, and utilize them to establish three types of quantum walks, denoted as $U^{up}$, $U^{down}$, and $U^{h}$, respectively.
Secondly, we demonstrate the utility of these quantum walks in encoding ${\Delta}^{up}_k$, ${\Delta}^{down}_k$, and ${\Delta}_k$.
Finally, quantum singular value transformations are applied to the above Laplacians, resulting in the projectors in our main result.

\subsection{Formal statements of quantum walk based projectors}\label{sec::Formal statements of main results}

We formally present three main results, each of which is derived through the utilization of the three aforementioned quantum walks.
Recall that we say $U$ is a {$\Pi_s$-projected $(a,\alpha,\epsilon)$-unitary encoding of $A$} if 
    $$
    \|
    A-a (\Pi_s \otimes \bra{0^\alpha}) U (\Pi_s \otimes \ket{0^\alpha} )
    \|
    \leq \epsilon.
    $$

\begin{theorem}
\label{thm:main}
Given a simplicial complex $X$ over $n$-vertices, 
we can implement $U_{Z^k}$, $U_{B_k}$, $U_{Z_k}$,  $U_{B^k}$,  $U_{\mathcal{H}_k}$ that are
$\Pi_k$-projected  $(1,\alpha,\epsilon)$-unitary encodings of 
$\mathrm{Proj}(Z^k)$, $\mathrm{Proj}(B_k)$, $\mathrm{Proj}(Z_k)$, $\mathrm{Proj}(B^k)$, $\mathrm{Proj}(\mathcal{H}_k)$, respectively,
with
\begin{itemize}
\item $\mathcal{O}(K^{up}\log(1/\epsilon)/\lambda_k^{up})$-number of use of 
$U^{up}$, $(U^{up})^\dagger$, C$_\Pi$NOT and single qubit gates for $U_{Z^k}$ and $U_{B_k}$, 
\item
$\mathcal{O}(K^{down}\log(1/\epsilon)/\lambda_k^{down})$-number of use of 
$U^{down}$, $(U^{down})^\dagger$, C$_\Pi$NOT and single qubit gates for $U_{Z_k}$ and $U_{B^k}$, 
\item
$\mathcal{O}(K\log(1/\epsilon)/\lambda_k)$-number of use of 
$U^h$, $(U^h)^\dagger$, C$_\Pi$NOT and single qubit gates for $U_{\mathcal{H}_k}$,
\end{itemize}
where
$\Pi= \Pi_k\otimes \ket{0^{\alpha}}\bra{0^{\alpha}}$,
and 
$U^{up}$, $U^{down}$, $U^{h}$
are quantum walk unitaries
that are $\Pi^\pm_k$-projected $(1,\alpha,0)$-unitary encodings of $\Pi^\pm_k P^{up} \Pi^\pm_k$, $\Pi^\pm_k P^{down} \Pi^\pm_k$, $\Pi^\pm_k P \Pi^\pm_k$, respectively. 
Here,
$\lambda_k^{up}$, $\lambda_k^{down}$, $\lambda_k$ are the smallest non-zero eigenvalues of $\Delta_k^{up}$, $\Delta_k^{down}$, $\Delta_k$, respectively.  
\end{theorem}

In the above theorem, unitary encoding of $\mathrm{Proj}(\mathcal{H}_k)$ is implemented using the harmonic quantum walk unitary $U^h$. 
As a corollary of Theorem~\ref{thm:main}, 
we can alternatively implement a projector onto $\mathcal{H}_k$ by combining the up quantum walk $U^{up}$ and the down quantum walk $U^{down}$.

\begin{corollary}
\label{thm:proj_harmonics_alternative}
Given a simplicial complex $X$ over $n$-vertices, 
we can implement $U_{\mathcal{H}_k}$  
such that 
\begin{equation}
\label{eq:Ud}
\left(\Pi_k\otimes \bra{0^{2\alpha}}\right) U_{\mathcal{H}_k} \left(\Pi_k\otimes \ket{0^{2\alpha}}\right)=\tilde{\Pi}_{\mathcal{H}_k}
\end{equation}
    where
$$
    \|\mathrm{Proj}(\mathcal{H}_k) -\tilde{\Pi}_{\mathcal{H}_k}\|\leq \epsilon,
$$
with 
$\mathcal{O}(K^{down}\log(1/\epsilon)/\lambda_k^{down}+K^{up}\log(1/\epsilon)/\lambda_k^{up})$-number of use of 
$U^{down}$, $(U^{down})^\dagger$, $U^{up}$, $(U^{up})^\dagger$, C$_\Pi$NOT and single qubit gates with $\Pi= \Pi_k\otimes \ket{0^{\alpha}}\bra{0^{\alpha}}$.
\end{corollary}

\begin{proof}
    This result is based on taking the product of projected unitary encoding of projectors onto $k$-cycle and $k$-cocycle (Theorem~\ref{thm:main}).  
    Note that $\mathcal{H}_k=Z_k\cap Z^k$. 
    Using the method of taking the product of block-encoded matrices of~\cite{gilyen2019quantum}, the theorem follows. 
\end{proof}

\subsection{Proof of Theorem~\ref{thm:main}}
\label{sec::Proofs outline}

We complete the proofs of our theorems in the subsequent subsections as follows. 
In Section~\ref{sec:def_up_random_walk},
we define the Markov transition matrix $P^{up}$ and show an encoding of $\Delta^{up}_k$ into up quantum walk.
In Section~\ref{sec:def_down_random_walk},
we define the Markov transition matrix $P^{down}$ and show an encoding of $\Delta^{down}_k$ into up quantum walk.
In Section~\ref{sec:def_harmonic_markov_chain},
we define the Markov transition matrix $P$ and show an encoding of $\Delta_k$ into harmonic quantum walk.
In Section~\ref{sec:proof_projectors},
we show how to approximately implement the encoding of projectors and show the complexity of the procedures.

\subsubsection{Up quantum walk and encoding of up Laplacian}
\label{sec:def_up_random_walk}

We start by formally defining the up random walk on simplicial complexes. 
Let us introduce the following normalization factor for up random walk:
\begin{equation}
\label{eq:normalization_up}
K^{up} \coloneqq \max_{\sigma\in X^\pm_k}  \sum_{\substack{i\in[n-k-1]\\j\in [k+2]}} w(i(\sigma))w(\tilde{j}(\sigma^\uparrow(i))).
\end{equation}
Note that $K^{up}\in poly(n)$ as long as the $w(v)\in poly(n)$ for all $v\in V$. 
Let us also define
$$
\eta^{up}_\sigma \coloneqq 1- \frac{1}{K^{up}}\sum_{\sigma'\in X^+_k}\left|(\Delta_k^{up} )_{\sigma',\sigma_+}\right|.
$$

\begin{definition}[Up random walk on simplicial complexes]
    \label{def:up_random_walk}

    We define a Markov chain transition matrix $P^{up}:S_k\rightarrow S_k$ for up random walk as follows
\begin{align}\label{eq:Markov_chain_up}
(P^{up})_{\sigma,\sigma'}&=
    \begin{cases}
    {\sum_{u \in up(\sigma)} w(u)^2 }/{K^{up}}
        , & if\ \sigma\neq \Theta\ and\ {\sigma=\overline{\sigma'}} \\
        w(v_\sigma)w(v_{\sigma'})/K^{up}, &if\ \sigma\neq \Theta,\ \ \sigma
        \sim_{\uparrow}
        \sigma'
        \\
        \eta_{\sigma}^{up} & if\ \sigma\neq \Theta\ and\ \sigma'=\Theta \\
        1, & if\ \sigma=\sigma'=\Theta \ \\
        0, & otherwise.
    \end{cases}
\end{align}
\end{definition}

This is a random walk on oriented $k$-simplices and an additional absorbing state $\Theta$. A $k$-simplex moves to up-adjacent simplex $\sigma'$ with probability $w(v_\sigma)w(v_{\sigma'})/K^{up}$. The lazyness for $\sigma$ is given by $\sum_{u \in up(\sigma)} w(u)^2/{K^{up}}$ . For unweighted simplicial complexes, $\sum_{u \in up(\sigma)} w(u)^2$ is equivalent to the up-degree of $\sigma$.

Recall that the matrix representation of up Laplacian is
 \begin{equation*}
    (\Delta^{up}_k)_{\sigma,\sigma'}= \begin{cases}
        \sum_{u \in up(\sigma)} w(u)^2, & if\ \sigma=\sigma' \\
        w(v_{\sigma})w(v_{\sigma'}), & if\ \sigma\sim_{\uparrow}{\sigma'}\\
        -w(v_{\sigma})w(v_{\sigma'}), & if\ \sigma\sim_{\uparrow}\overline{\sigma'}\\
        0, & otherwise.
    \end{cases}
    \end{equation*}
From the matrix representation of the up Laplacian, it can be seen that 
\begin{equation}
\label{eq:up_laplacian_markovcain}
    (P^{up})_{\sigma,\sigma'}=
    \begin{cases} ({\Delta}_k^{up})_{\sigma_+,\sigma_+}/K^{up}, & if\ \sigma\neq \Theta\ and\ {\sigma=\overline{\sigma'}} \\
        |({\Delta}_k^{up})_{\sigma_+',\sigma_+}|/K^{up}, &if\ \sigma\neq \Theta,\ \ \sigma
        \sim_{\uparrow}
        \sigma'
        \\
        \eta_{\sigma}^{up} & if\ \sigma\neq \Theta\ and\ \sigma'=\Theta \\
        1, & if\ \sigma=\sigma'=\Theta \ \\
        0, & otherwise.
    \end{cases}
\end{equation}

We then consider the quantization of the up random walk. 
Suppose we have a projected unitary encoding of $\Pi^\pm_k P^{up} \Pi^\pm_k$, then we can construct an encoding of the up Laplacian by taking the coherent interference between the positively oriented simplices and corresponding negatively oriented simplices after the application of a quantum walk step:

\begin{prop}
\label{prop:unitaryencoding_up}
Given a simplicial complex $X$ over $n$-vertices,  
 let $U^{up}$ be a $\Pi^\pm_k$-projected $(1,\alpha,0)$-unitary encoding of $\Pi^\pm_k P^{up} \Pi^\pm_k$ i.e.,
   $$
    \left(\Pi^\pm_{k}\otimes\bra{0^{\alpha}}\right)
    U^{up} \left(\Pi^\pm_{k}\otimes \ket{0^{\alpha}}\right)= \Pi^\pm_{k}P^{up}\Pi^\pm_{k}.
    $$
    Then, $V^{up}\coloneqq  \left((I_n\otimes HZ\otimes I)\otimes I_{\alpha}\right)U^{up}$ 
    satisfies
     $$
    \left(\Pi_k\otimes\bra{0^{\alpha}}\right)
    V^{up} \left(\Pi_k\otimes \ket{0^{\alpha}}\right)= \frac{{\Delta}_k^{up}}{K^{up}\sqrt{2}}.
$$ 
\end{prop}

\begin{proof}

For $\sigma^+_x,\sigma^+_y \in X^+_k$,  
let us denote the corresponding unoriented simplices by $x,y \in \{0,1\}^n$, respectively.
Then, for all $\sigma^+_x,\sigma^+_y \in X^+_k$,
    \begin{align*}
    (\bra{\sigma^+_y}\otimes \bra{0^{\alpha}}) V^{up} (\ket{\sigma^+_x}\otimes \ket{0^{\alpha}})&= 
     (\bra{y}\bra{00}\otimes \bra{0^{n+2}}) V^{up} (\ket{x}\ket{00}\otimes \ket{0^{n+2}})\\
    & = 
    \left(\bra{y}\frac{\bra{00}-\bra{10}}{\sqrt{2}}\otimes \bra{0^{n+2}}\right) U^{up} \left(\ket{x}\ket{00}\otimes \ket{0^{n+2}}\right)\\
    & =
    \left(\frac{\bra{\sigma_y^+}-\bra{\sigma_y^-}}{\sqrt{2}}\otimes \bra{0^{n+2}}\right) U^{up}(\ket{\sigma_x^+}\otimes \ket{0^{n+2}})\\
    & = 
    \frac{1}{\sqrt{2}}\left(
    {P^{up}_{\sigma_x^+,\sigma_y^+}}
    -{P^{up}_{\sigma_x^+,\sigma_y^-}}
    \right)
    \\&=
    \begin{cases}
       \frac{1}{\sqrt{2}}P^{up}_{\sigma_x^+,\sigma_y^+}, & if\ \sigma_x^+=\sigma_y^+ \\
        \frac{1}{\sqrt{2}}P^{up}_{\sigma_x^+,\sigma_y^+}, &if\ \sigma_x^+\sim_{\uparrow}\sigma_y^+\\
       - \frac{1}{\sqrt{2}}P^{up}_{\sigma_x^+,\sigma_y^-}, &if\ \sigma_x^+\sim_{\uparrow}\sigma_y^-
        \\
        0, & otherwise.
    \end{cases}
\end{align*}
Using the relationship between the Markov transition matrix and the up Laplacian of eq.~\eqref{eq:up_laplacian_markovcain}, we can further compute
\begin{align*}
   (\bra{\sigma^+_y}\otimes \bra{0^{\alpha}}) V^{up} (\ket{\sigma^+_x}\otimes \ket{0^{\alpha}}) &=
    \begin{cases}
       \frac{1}{K^{up}\sqrt{2}}({\Delta}^{up}_k)_{\sigma_x^+,\sigma_x^+}, & if\ \sigma_x^+=\sigma_y^+ \\
        \frac{1}{K^{up}\sqrt{2}}|({\Delta}^{up}_k)_{\sigma_x^+,\sigma_y^+}|, &if\ \sigma_x^+\sim_{\uparrow}\sigma_y^+\\
        -\frac{1}{K^{up}\sqrt{2}}|({\Delta}^{up}_k)_{\sigma_x^+,\sigma_y^+}|, &if\ \sigma_x^+\sim_{\uparrow}\sigma_y^-
        \\
        0, & otherwise.
    \end{cases}
    \\&=
    \begin{cases}
        \frac{1}{K^{up}\sqrt{2}}({\Delta}^{up}_k)_{\sigma_x^+,\sigma_x^+}, & if\ \sigma_x^+=\sigma_y^+ \\
        \frac{1}{K^{up}\sqrt{2}}({\Delta}^{up}_k)_{\sigma_x^+,\sigma_y^+}, &if\ \sigma_x^+\sim_{\uparrow}\sigma_y^+\\
        \frac{1}{K^{up}\sqrt{2}}({\Delta}^{up}_k)_{\sigma_x^+,\sigma_y^+}, &if\ \sigma_x^+\sim_{\uparrow}\sigma_y^-
        \\
        0, & otherwise.
    \end{cases}
    \end{align*}

   Therefore, it can be observed that
   \begin{equation}
    \left(\Pi_k \otimes\bra{0^{\alpha}}\right)
    V^{up} \left(\Pi_k\otimes \ket{0^{\alpha}}\right)= \frac{1}{K^{up}\sqrt{2}}{\Delta}^{up}_k.
    \end{equation}
\end{proof}

\begin{remark}
\label{remark:Pup}
    There could be several ways of realizing the unitary encoding of $\Pi^\pm_k P^{up} \Pi^\pm_k$. 
For example, based on the Szegedy type quantum walk, we can consider the following unitary that acts on two state registers for all $\sigma\in S_k$: 
    \begin{equation}
    \label{eq:Q_walk_up}
    U: \ket{\sigma}\ket{0^{n+2}}\rightarrow
    \sum_{\sigma'\in S} \sqrt{P^{up}_{\sigma,\sigma'}} \ket{\sigma}\ket{\sigma'}.
    \end{equation}
Then, we can construct a projected unitary encoding of $P^{up}$ by $U^\dagger \cdot \mathrm{SWAP}\cdot  U$:

$$
\left(\Pi^\pm_{k}\otimes\bra{0^{n+2}}\right)
    U^\dagger \cdot \mathrm{SWAP}\cdot  U \left(\Pi^\pm_{k}\otimes \ket{0^{n+2}}\right)= \Pi^\pm_k P^{up} \Pi^\pm_k.
$$
This is because   for all $\sigma_x,\sigma_y \in S_k$,
    \begin{align*}
        (\bra{\sigma_y}\otimes \bra{0^{n+2}}) U^\dagger \cdot \mathrm{SWAP}\cdot  U (\ket{\sigma_x}\otimes \ket{0^{n+2}})
    &=\left(
       \sum_{\sigma_y'\in S_k} \sqrt{P^{up}_{\sigma_y,\sigma_y'}} \bra{\sigma_y}\bra{\sigma_y'}
      \right) 
      \left(\sum_{\sigma_x'\in S_k} \sqrt{P^{up}_{\sigma_x,\sigma_x'}} \ket{\sigma_x'}\ket{\sigma_x}\right)\\
    &= \sqrt{P^{up}_{\sigma_x,\sigma_y}}\sqrt{P^{up}_{\sigma_y,\sigma_x}},
    \end{align*}
and for all $\sigma_x,\sigma_y \in X^\pm_k$, 
$$  (\bra{\sigma_y}\otimes \bra{0^{n+2}}) U^\dagger \cdot \mathrm{SWAP}\cdot  U(\ket{\sigma_x}\otimes \ket{0^{n+2}})=\sqrt{P^{up}_{\sigma_x,\sigma_y}}\sqrt{P^{up}_{\sigma_y,\sigma_x}}= P^{up}_{\sigma_x,\sigma_y} $$
because the Markov transition matrix is symmetric except for the elements related to the absorbing state $\Theta$. 
In Section~\ref{sec:clique_construction} we give explicit construction of the unitary encoding of Markov chain for clique complexes. 
\end{remark}

\subsubsection{Down quantum walk and encoding of down Laplacian}
\label{sec:def_down_random_walk}

Let us introduce the following normalization factor for a down random walk:
\begin{equation}
\label{eq:normalization_down}
K^{down} \coloneqq\max_{\sigma \in X^\pm_k}\sum_{i\in[k+1],j\in [n-k]} w(\tilde{i}(\sigma))w(\tilde{j}(\sigma^\downarrow(i))).
\end{equation}
Note that $K^{down}\in poly(n)$ as long as the $w(v)\in poly(n)$ for all $v\in V$. 
Let us also define
$$
\eta^{down}_\sigma \coloneqq 1- \frac{1}{K^{down}}\sum_{\sigma'\in X^+_k}\left|(\Delta_k^{down} )_{\sigma',\sigma_+}\right|.
$$

\begin{definition}[Down random walk on simplicial complexes]
    \label{def:down_random_walk}
    We define a Markov chain transition matrix $P^{down}:S_k\rightarrow S_k$ for a down random walk as follows:
\begin{equation}\label{eq:Markov_chain_down}
(P)^{down}_{\sigma,\sigma'}=
    \begin{cases}
        \sum_{v \in \sigma} w(v)^2/K^{down}, & if\ \sigma\neq \Theta\ and\ {\sigma={\sigma'}} \\
        w(v_{\sigma})w(v_{\sigma'})/K^{down}, &if\ \sigma\neq \Theta,\ \ \sigma
        \sim_{\downarrow}
        \sigma'
        \\
        \eta_{\sigma}^{down} & if\ \sigma\neq \Theta\ and\ \sigma'=\Theta \\
        1, & if\ \sigma=\sigma'=\Theta \ \\
        0, & otherwise.
    \end{cases}
\end{equation}
\end{definition}

Recall that the matrix representation of the down Laplacian for $\sigma, \sigma' \in X^+_k$ is 
\begin{equation}
    (\Delta^{down}_k)_{\sigma,\sigma'}=
    \begin{cases}
        \sum_{v \in \sigma} w(v)^2, & if\ \sigma=\sigma' \\
        w(v_{\sigma})w(v_{\sigma'}), & if\ \sigma\sim_{\downarrow}{\sigma'}\\
        -w(v_{\sigma})w(v_{\sigma'}), & if\ \sigma\sim_{\downarrow}\overline{\sigma'}\\
        0, & otherwise. 
    \end{cases}
\end{equation}

Then, it can be seen that 

$$
    (P^{down})_{\sigma,\sigma'}=
    \begin{cases} ({\Delta}_k^{down})_{\sigma_+,\sigma_+}/K^{down}, & if\ \sigma\neq \Theta\ and\ {\sigma=\overline{\sigma'}} \\
        |({\Delta}_k^{down})_{\sigma_+',\sigma_+}|/K^{down}, &if\ \sigma\neq \Theta,\ \ \sigma
        \sim_{\downarrow}
        \sigma'
        \\
        \eta_{\sigma}^{down}, & if\ \sigma\neq \Theta\ and\ \sigma'=\Theta \\
        1, & if\ \sigma=\sigma'=\Theta \ \\
        0, & otherwise.
    \end{cases}
$$

With a projected unitary encoding of $P^{down}$, we can construct a unitary in which down Laplacian is encoded:

\begin{prop}
\label{prop:unitaryencoding_down}

Given a simplicial complex $X$ over $n$-vertices, 
let $U^{down}$ be a $\Pi^\pm_k$-projected $(1,\alpha,0)$-unitary encoding of $\Pi^\pm_k P^{down} \Pi^\pm_k$ i.e.,
    $$
\left(\Pi^\pm_k\otimes\bra{0^{\alpha}}\right)
    U^{down} \left(\Pi^\pm_k\otimes \ket{0^{\alpha}}\right)=\Pi^\pm_k P^{down} \Pi^\pm_k.
    $$
    Then, $V^{down}\coloneqq  \left((I_n\otimes HZ\otimes I)\otimes I_{\alpha}\right)U^{down}$ satisfies
     $$
    \left(\Pi_k\otimes\bra{0^{\alpha}}\right)
    V^{down} \left(\Pi_k\otimes \ket{0^{\alpha}}\right)= \frac{{\Delta}_k^{down}}{K^{down}\sqrt{2}}.
$$ 
\end{prop}
\begin{proof}
    The proof is almost similar to the proof of Proposition~\ref{prop:unitaryencoding_up}. With a similar calculation, we can show that 
        \begin{align*}
    (\bra{\sigma^+_y}\otimes \bra{0^{\alpha}}) V (\ket{\sigma^+_x}\otimes \ket{0^{\alpha}})&= 
    \frac{1}{\sqrt{2}}\left(
    {P^{down}_{\sigma_x^+,\sigma_y^+}}
    -{P^{down}_{\sigma_x^+,\sigma_y^-}}
    \right)
    \\&=
    \begin{cases}
       \frac{1}{\sqrt{2}}P^{down}_{\sigma_x^+,\sigma_y^+}, & if\ \sigma_x^+=\sigma_y^+ \\
        \frac{1}{\sqrt{2}}P^{down}_{\sigma_x^+,\sigma_y^+}, &if\ \sigma_x^+\sim_{\downarrow}\sigma_y^+\\
       - \frac{1}{\sqrt{2}}P^{down}_{\sigma_x^+,\sigma_y^-}, &if\ \sigma_x^+\sim_{\downarrow}\sigma_y^-
        \\
        0, & otherwise
    \end{cases}
    \\&=
    \begin{cases}
        \frac{1}{K^{down}\sqrt{2}}({\Delta}^{down}_k)_{\sigma_x^+,\sigma_x^+}, & if\ \sigma_x^+=\sigma_y^+ \\
        \frac{1}{K^{down}\sqrt{2}}({\Delta}^{down}_k)_{\sigma_x^+,\sigma_y^+}, &if\ \sigma_x^+\sim_{\downarrow}\sigma_y^+\\
        \frac{1}{K^{down}\sqrt{2}}({\Delta}^{down}_k)_{\sigma_x^+,\sigma_y^+}, &if\ \sigma_x^+\sim_{\downarrow}\sigma_y^-
        \\
        0, & otherwise.
    \end{cases}
    \end{align*}
\end{proof}

\subsubsection{Harmonic quantum walk and encoding of full Laplacian}
\label{sec:def_harmonic_markov_chain}
In this subsection, we introduce the random walk on simplicial complexes that we will quantize in the subsequent section. 

Given a directed simplicial complex $X^\pm$,
define 

\begin{equation}
\label{eq:normalization_harmonic}
K \coloneqq \max_{\sigma\in X^\pm_k} 
\left(
\sum_{i\in [n]} w(i)^2
+
\sum_{\substack{i\in [n-k-1],\\j\in [k+1]}} w(i(\sigma))w(\tilde{j}(\sigma))
\right).
\end{equation}
Note that $K\in poly(n)$ as long as the $w(v)\in poly(n)$ for all $v\in V$. 
Then, let us also introduce 
$$
\eta_\sigma \coloneqq 1- \frac{1}{K}\sum_{\sigma'\in X^+_k}\left|(\Delta_k )_{\sigma',\sigma_+}\right|.
$$

\begin{definition}[Harmonic random walk on simplicial complexes]
    \label{def:harmonic_random_walk}

We define a Markov chain transition matrix $P:S_k\rightarrow S_k$ as follows
\begin{align}\label{eq:Markov_chain}
(P)_{\sigma,\sigma'}&=
    \begin{cases}
      {\left(\sum_{u \in up(\sigma)} w(u)^2+\sum_{v \in \sigma} w(v)^2\right)}/{K}, & if\ \sigma\neq \Theta\ and\ {\sigma={\sigma'}} \\
        w(v_{\sigma})w(v_{\sigma'})/K, &if\ \sigma\neq \Theta,\ \sigma
        \sim_{\downarrow}
        \sigma' \ and\ \sigma\nsim_{\uparrow} \overline{\sigma'}
        \\
        \eta_{\sigma} & if\ \sigma\neq \Theta\ and\ \sigma'=\Theta \\
        1, & if\ \sigma=\sigma'=\Theta \ \\
        0, & otherwise.
    \end{cases}
\end{align}
\end{definition}

It can be seen that this is a stochastic matrix, i.e., $\sum_{\sigma'\in S_k} (P)_{\sigma,\sigma'}=1$ for all $\sigma\in S_k$.
Recalling the matrix representation of the combinatorial Laplacian 
    \begin{equation}
    (\Delta_k)_{\sigma,\sigma'}=
    \begin{cases}
        \deg(\sigma)+k+1, & if\ \sigma=\sigma' \\
        1, &if\ \sigma\sim_{\downarrow}{\sigma'}\ and\  \sigma\nsim_{\uparrow}\overline{\sigma'}
        \\
        -1, &if\ \sigma\sim_{\downarrow}\overline{\sigma'}\ and\ \sigma\nsim_{\uparrow}\sigma' \\
        0, & otherwise,\ 
    \end{cases}
    \end{equation}
it can be seen that
\begin{align*}
(P)_{\sigma,\sigma'}=
    \begin{cases}
        ({\Delta}_k)_{\sigma_+,\sigma_+}/K , & if\ \sigma\neq \Theta\ and\ {\sigma=\sigma'}\\
        |({\Delta}_k)_{\sigma_+',\sigma_+}|/K, &if\ \sigma\neq \Theta,\ \sigma
        \sim_{\downarrow}
        \sigma' \ and\ \sigma\nsim_{\uparrow} \overline{\sigma'}
        \\
        \eta_{\sigma} & if\ \sigma\neq \Theta\ and\ \sigma'=\Theta \\
        1, & if\ \sigma=\sigma'=\Theta \ \\
        0, & otherwise.
    \end{cases}
\end{align*}

Based on this definition, we show the following proposition:

{

\begin{prop}
\label{prop:unitaryencoding}
Given a simplicial complex $X$ over $n$-vertices, 
suppose that $U^h$ is a $\Pi^\pm_k$-projected $(1,\alpha,0)$-unitary encoding of $\Pi^\pm_k P \Pi^\pm_k$ i.e.,
    $$
    \left(\Pi^\pm_k\otimes\bra{0^{\alpha}}\right)
    U^{h} \left(\Pi^\pm_k\otimes \ket{0^{\alpha}}\right)=\Pi^\pm_k P \Pi^\pm_k.
    $$
Let $V^h= \left((I_n\otimes HZ\otimes I)\otimes I_{\alpha}\right)((U^h)^\dagger \cdot \mathrm{SWAP} \cdot U^h)$. 
Then, $V^h$ satisfies
 $$
    \left(\Pi_k\otimes\bra{0^{\alpha}}\right)
    V^h \left(\Pi_k\otimes \ket{0^{\alpha}}\right)= \frac{{\Delta}_k}{K\sqrt{2}}.
    $$ 
\end{prop}}

\begin{proof}
The proof is similar to the proof of Proposition~\ref{prop:unitaryencoding_up}.
In this case, for all $\sigma^+_x,\sigma^+_y \in X^+_k$,
    \begin{align*}
    (\bra{\sigma^+_y}\otimes \bra{0^{\alpha}}) V^{h} (\ket{\sigma^+_x}\otimes \ket{0^{\alpha}})
    & = 
    \frac{1}{\sqrt{2}}\left(
    {P_{\sigma_x^+,\sigma_y^+}}
    -{P_{\sigma_x^+,\sigma_y^-}}
    \right)
    \\&=
    \begin{cases}
       \frac{1}{K\sqrt{2}}({\Delta}_k)_{\sigma_x^+,\sigma_x^+}, & if\ \sigma_x^+=\sigma_y^+ \\
        0, &if\ \sigma_x^+\sim_{\downarrow}\sigma_y^+\ and\ \sigma_x^+\sim_{\uparrow}\sigma_y^-\\
       0, &if\ \sigma_x^+\sim_{\downarrow}\sigma_y^-\ and\ \sigma_x^+\sim_{\uparrow}\sigma_y^+
        \\
        \frac{1}{\sqrt{2}}P_{\sigma_x^+,\sigma_y^+}, &if\ \sigma_x^+\sim_{\downarrow}\sigma_y^+\ and\ \sigma_x^+\nsim_{\uparrow}\sigma_y^-\\
       - \frac{1}{\sqrt{2}}P_{\sigma_x^+,\sigma_y^-}, &if\ \sigma_x^+\sim_{\downarrow}\sigma_y^-\ and\ \sigma_x^+\nsim_{\uparrow}\sigma_y^+
        \\
        0, & otherwise.
    \end{cases}
    \end{align*}
Then, we can conclude the proof of the proposition.

\end{proof}

\begin{remark}

    We can give a different definition of harmonic random walk on simplicial complexes as follows:

\begin{align}
(P')_{\sigma,\sigma'}&=
    \begin{cases}
      {\left(\sum_{u \in up(\sigma)} w(u)^2+\sum_{v \in \sigma} w(v)^2\right)}/{K}, & if\ \sigma\neq \Theta\ and\ {\sigma={\sigma'}} \\
        w(v_{\sigma})w(v_{\sigma'})/K, &if\ \sigma\neq \Theta,\ \sigma
        \sim_{\downarrow}
        \sigma' \ or\ \sigma\sim_{\uparrow} \sigma'
        \\
        \eta_{\sigma}, & if\ \sigma\neq \Theta\ and\ \sigma'=\Theta \\
        1, & if\ \sigma=\sigma'=\Theta \ \\
        0, & otherwise.
    \end{cases}
\end{align}
Note that $\sigma\sim_{\downarrow} \sigma'$ and $\sigma\sim_{\uparrow} \sigma'$ do not coincide. 
Although this definition also works for the quantization and quantum encoding of the Laplacian, we mainly consider the definition in Section~\ref{sec:def_harmonic_markov_chain} in this paper.
\end{remark}

\subsubsection{Implementing the unitary encoding of projectors}
\label{sec:proof_projectors}

In this subsection, we finally show how to implement the unitary encoding of the projector onto the $k$-cocycle $Z^k=\ker(\partial^*_{k+1})=\ker(\Delta^{up}_k)$ of vertex-weighted simplicial complexes using up quantum walk. 

The projector can be implemented by applying the quantum singular value transformation (QSVT)~\cite{gilyen2019quantum} to $V^{up}$ of Proposition~\ref{prop:unitaryencoding_up}, which is a unitary encoding of the up Laplacian. 
We use the following polynomial, which is an approximation of the rectangle function, to transform the eigenvalues.

\begin{lemma}[Lemma 29, \cite{gilyen2019quantum}]
\label{lemma:rectangle}
Let $\delta,\epsilon\in(0,1/2)$ and $t\in [-1,1]$. There exists an even polynomial $P\in {\mathbb R}[x]$ of degree ${\mathcal O}(\log{(1/\epsilon)}/\delta)$, such that $|P(x)|\leq 1$ for all $x\in[-1,1]$, and 
\begin{equation}
\left\{
    \begin{array}{ll}
    P(x)\in [0,\epsilon], & \forall x\in[-1,-t-\delta]\cup [t+\delta,1],\ and\\
    P(x)\in [1-\epsilon,1], & \forall x\in [-t+\delta,t-\delta].
    \end{array}
\right.
\end{equation}
\end{lemma}

Let $\tilde{\lambda}_k^{up}\coloneqq \frac{\lambda^{up}_k}{\sqrt{2}K^{up}}$ where $\lambda^{up}_k$ is the smallest non-zero eigenvalue of $\Delta^{up}_k$.
We perform QSVT with this polynomial by taking $t=\tilde{\lambda}_k^{up}/2$ and $\delta=\tilde{\lambda}_k^{up}/4$.
Then, the degree of $P(x)$ is $T= \mathcal{O}(K^{up}\log(1/\epsilon)/\lambda^{up}_k)$. 
By performing QSVT for $V^{up}$, 
we can implement $V^{up}_T$ s.t. 
    \begin{align}
    \left(\Pi_k\otimes \bra{00}\otimes \bra{0^{\alpha}}\right) V^{up}_T \left(\Pi_k\otimes \ket{00}\otimes \ket{0^{\alpha}}\right)
    &=
    P\left( 
    \frac{\Delta^{up}_k}{\sqrt{2}K}
    \right)\\
    &=
    \sum_i
    P\left(\frac{\lambda^{up}_i}{\sqrt{2}K}\right) 
    \ket{\psi_i}\bra{\psi_i}
    \\
    &=:
    \tilde{\Pi}_{Z^k},
    \end{align}
    where $\{\lambda^{up}_i,\ket{\psi_i}\}_i$ are the eigenvalues and eigenvectors of $\Delta^{up}_k$.
    Then, $\tilde{\Pi}_{Z^k}$ satisfies 
    $$
    \|\mathrm{Proj}({Z^k})-\tilde{\Pi}_{Z^k}\|\leq \epsilon.
    $$
The unitary $V^{up}_T$ can be implemented using $T= \mathcal{O}(K\log(1/\epsilon)/\lambda_k)$ number of 
 $V^{up},(V^{up})^\dagger, C_\Pi NOT$ gates and other single qubit gates.
With up quantum walk, we can also implement the unitary encoding of the projector onto the $k$-boundary $B_k= \mathrm{Im}(\partial_{k+1})=\mathrm{Im}(\Delta^{up}_k)$. 
 In this case, instead of using the polynomial of Lemma~\ref{lemma:rectangle}, 
    we use a polynomial $P(x)$ s.t. 
\begin{equation}
\left\{
    \begin{array}{ll}
    P(x)\in [1-\epsilon,1], & \forall x\in[-1,-t-\delta]\cup [t+\delta,1],\ and\\
    P(x)\in [0,\epsilon], & \forall x\in [-t+\delta,t-\delta].
    \end{array}
\right.
\end{equation}
of \cite{martyn2021grand}. 
By performing QSVT with this polynomial, we can implement $U_{B_k}$ with the desired complexity. 

With down quantum walks, we can similarly implement projectors onto $k$-cycle $Z_k=\ker(\partial_k)=\ker(\Delta^{down}_k)$ and $k$-coboundary $B^k=\mathrm{Im}(\partial_k^*)$.
Finally, we can implement a projector onto the harmonics with the harmonic quantum walk in a similar way as in the case of the up walk and down walk. 
This concludes the proof of Theorem~\ref{thm:main}.

\section{Efficient construction of quantum walk unitaries for clique complexes}

\label{sec:clique_construction}
\input{clique_construction}

\section{Applications to topological data analysis problems}
\label{sec:applications}

\input{application}

{

\section{Application to the High-Dimensional Discrete Dirichlet Problem}
\label{sec:dirichlet}

In this section, we consider an application of our quantum walk framework to the high-dimensional discrete Dirichlet problem (HDDP) introduced by Rosenthal~\cite{rosenthal2014simplicial} and discuss potential quantum advantage for this problem.
 
\subsection{Problem Setup}

The Dirichlet problem is a fundamental problem in mathematics and physics, with a rich history dating back to the 19th century.
In the discrete setting, the Dirichlet problem on a graph asks for a harmonic function on the interior vertices of a graph given prescribed values on the boundary vertices, and arises naturally in a wide range of applications including electrical network analysis, random walks, and machine learning~\cite{doyle1984random}.

HDDP~\cite{rosenthal2014simplicial} generalizes this classical problem to the setting of simplicial complexes, where the unknown function is a $k$-cochain rather than a function on vertices.
This generalization is significant because it provides a natural framework for studying harmonic $k$-forms with prescribed boundary conditions, connecting spectral geometry, Hodge theory, and combinatorial topology.

The high-dimensional discrete Dirichlet problem is defined as follows.
 
\begin{definition}[High-Dimensional Discrete Dirichlet Problem (HDDP)]
Let $X = X(G)$ be the clique complex of a graph $G$ on $n$ vertices, and 
let $X_k = A \sqcup (X_k \setminus A)$ be a partition into \emph{boundary} simplices $A$ and \emph{interior} simplices $X_k \setminus A$, and let $f \in C^k(A;\mathbb{R})$ be a boundary cochain.
We assume the partition $X_k = A \sqcup (X_k \setminus A)$ and the boundary cochain $f$ are described with $\mathrm{poly(n)}$ bit strings and we can efficiently decide $\sigma \in A$ or $\sigma \in X_k \setminus A$ for any $\sigma\in X_k$, and 
$$\ket{f}:= \frac{1}{\sqrt{\sum_{\sigma\in A} f^2(\sigma)}}\sum_{\sigma \in A} f(\sigma) \ket{\sigma}$$
can be prepared by $\mathrm{poly}(n)$-size quantum circuit.
The high-dimensional discrete Dirichlet problem asks for a $k$-cochain $F \in C^k(X_k;\mathbb{R})$ satisfying
\begin{equation}
    \Delta_k^{up} F\big|_{X_k \setminus A} = 0, \qquad F\big|_A = f.
\end{equation}
\end{definition}

We note that, throughout this section, we restrict attention to real coefficients and unweighted simplicial complexes.
By~\cite[Theorem~5.1 and Claim~5.2]{rosenthal2014simplicial}, a unique solution exists whenever
\begin{equation}
\label{eq:acyclicity}
    H_k(X, A) = 0,
\end{equation}
where $ H_k(X, A)$ is the relative homology.

\subsection{Quantum Algorithm}
 
Under the partition $X_k = A \sqcup (X_k \setminus A)$, the up Laplacian $\Delta_k^{up}$ admits the block decomposition~\cite{rosenthal2014simplicial}
\begin{equation}
    \Delta_k^{up}= \begin{pmatrix} \Delta^{up}_A & -R \\ -Q & \Delta^{up}_{X_k \setminus A} \end{pmatrix},
\end{equation}
where $\Delta^{up}_{X_k \setminus A} := \Pi_{X_k \setminus A} \,\Delta_k^{up}\, \Pi_{X_k \setminus A} |_{C_k(X_k \setminus A)}$ is the \emph{Dirichlet Laplacian} and $\Pi_{X_k \setminus A}$ denotes the orthogonal projector onto the subspace spanned by interior simplices $X_k \setminus A$.
The Dirichlet condition $\Delta_k^{up} F|_{X_k \setminus A} = 0$ then reduces to the linear system
\begin{equation}
    \Delta^{up}_{X_k \setminus A}\, F_{X_k \setminus A} = Q f =: |r\rangle,
\end{equation}
where $F_{X_k \setminus A}$ denotes the restriction of $F$ to interior simplices.
Under the condition of eq.~\eqref{eq:acyclicity}, $\Delta^{up}_{X_k \setminus A}$ is positive definite on $C_k(X_k \setminus A)$, so the solution $F_{X_k \setminus A} = (\Delta^{up}_{X_k \setminus A})^{-1} |r\rangle$ is unique.
 
A block encoding of $\Delta^{up}_{X_k \setminus A}$ follows from that of $\Delta_k^{up}$ (Proposition~\ref{prop:unitaryencoding_up}) by restricting the simplex register from $X^{k}$ to $X^{k} \setminus A$, which is a straightforward modification for clique complexes and the access to the given partition $X_k = A \sqcup (X_k \setminus A)$.
We then apply QSVT~\cite{gilyen2019quantum} with a polynomial approximation to $x^{-1}$ on the spectral interval $[\lambda_{\min}, \|\Delta^{up}_{X_k \setminus A}\|]$, where $\lambda_{\min} > 0$ is the smallest eigenvalue of $\Delta^{up}_{X_k \setminus A}$.
 
\begin{theorem}
\label{thm:dirichlet}
Let $X = X(G)$ be the clique complex of a graph $G$ on $n$ vertices, and suppose $H_k(X, A) = 0$.
Let $\lambda_{\min} > 0$ be the smallest eigenvalue of $\Delta^{up}_{X_k \setminus A}$, and suppose $\|Qf\| = \Omega(\|f\|)$.
Then there exists a quantum algorithm that outputs a state $\varepsilon$-close to
\begin{equation}
    |F_{X_k \setminus A}\rangle \propto (\Delta^{up}_{X_k \setminus A})^{-1}\, |r\rangle
\end{equation}
using
\begin{equation}
    O\!\left(\frac{\mathrm{poly}(n) \cdot \log(1/\varepsilon)}{\lambda_{\min}}\right)
\end{equation}
gates.
\end{theorem}

\begin{remark}
Because $f$ is known classically, the full solution 
$F \in C^k(X_k; \mathbb{R})$ is obtained by combining 
$F_{X_k \setminus A}$ with the boundary values $f$.
\end{remark}

\begin{remark}
The assumption $\|Qf\| = \Omega(\|f\|)$ indicates that
the boundary values $f$ have a non-trivial influence on the interior simplices. We believe this is a natural assumption. 
\end{remark}
 
\begin{proof}
The block encoding of $\Delta^{up}_{X_k \setminus A}$ is obtained from Proposition~\ref{prop:unitaryencoding_up} at cost $\mathrm{poly}(n)$.
The right-hand side $|r\rangle = Q|f\rangle = \Pi_{X_k \setminus A} \,\Delta_k^{up}\, \Pi_A |f\rangle$ is prepared by one application of the block encoding of $\Delta_k^{up}$ to $\Pi_A |f\rangle$, followed by the projector $\Pi_{X_k \setminus A}$, which is implementable in $O(n^2)$ gates for clique complexes. Under the assumption $\|Qf\| = \Omega(\|f\|)$, the post-projection state $|r\rangle$ is obtained with $\Omega(1)$ success probability. The total cost is $\mathrm{poly}(n)$ by assumption on $|f\rangle$.
We then apply QSVT~\cite{gilyen2019quantum} with a degree-$O(\log(1/\varepsilon)/\lambda_{\min})$ polynomial approximation to $x^{-1}$, each iteration using one call to the block encoding of $\Delta^{up}_{X_k \setminus A}$.
The total gate count is $O(\mathrm{poly}(n) \cdot \log(1/\varepsilon) / \lambda_{\min})$.
\end{proof}
 
\subsection{Quantum Advantage}
 
We compare the quantum complexity of Theorem~\ref{thm:dirichlet} with the best known classical approaches.
The space $C^k(X_k)$ has dimension $|X_k| = \binom{n}{k+1} = \Theta(n^{k+1})$.
Classical direct solvers for linear systems $\Delta^{up}_{X_k \setminus A}\, F_{X_k \setminus A} = |r\rangle$ therefore require $\Omega(n^{(k+1)\omega})$ operations, where $\omega \approx 2.37$ is the matrix multiplication exponent.
For non-constant $k$, this best known classical direct-solver approach takes superpolynomial time. Hence, the quantum algorithm of Theorem~\ref{thm:dirichlet} indicates an apparent superpolynomial speedup over the best known classical algorithm whenever $\lambda_{\min} = \Omega(1/\mathrm{poly}(n))$.
We leave showing a quantum speedup for HDDP under plausible computational complexity assumptions as an open problem.

We remark that the assumption $\|Qf\| = \Omega(\|f\|)$ is 
natural, since it holds whenever the boundary cochain $f$ is not concentrated on simplices that are decoupled from the interior.
When $\|Qf\|$ is exponentially small, the preparation of $|r\rangle$ requires exponentially many repetitions, and the complexity of the HDDP in this regime remains an open problem.
It would be interesting to characterize the boundary conditions $f$ for which $\|Qf\| = \Omega(\|f\|)$ holds, and to determine whether a quantum speedup persists in the regime $\|Qf\| = o(\|f\|)$.
}

\section*{Acknowledgement}
RH was supported by JSPS KAKENHI Grant Number JP22J11727 and PRESTO Grant Number JPMJPR23F9, Japan.

\bibliographystyle{alphaurl}
\bibliography{main}

\appendix

\section{Overview of the classical random walks on simplicial complexes}
\label{app:classical}

In this section, we give an overview for the definitions of up and down random walks of~\cite{parzanchevski2017simplicial}~\cite{mukherjee2016random}. 
The results related to these walks are summarized in Table~\ref{tab:classical}. 
Note that we use the notations we introduce in Section~\ref{sec:pre:adjacency}, that are sometimes slightly different from the notations of~\cite{parzanchevski2017simplicial}~\cite{mukherjee2016random}.

The Markov transition matrix of the up random walk of~\cite{parzanchevski2017simplicial} is defined as 
\begin{equation}
\label{NE:eq:up}
P_{\sigma,\sigma'}=
\begin{cases}
    p & \sigma'=\sigma \\
    \frac{1-p}{(k+1) \mathrm{deg}(\sigma)} & {\sigma'}\sim_{\uparrow}\overline{\sigma} \\
    0 & otherwise.
\end{cases}
\end{equation}
Here, $p$ is the parameter of laziness. 
Note that $\sigma$ transits to $\sigma'$ s.t. $\overline{\sigma}\sim_{\uparrow}$ in this definition while we transit to $\sigma'$ s.t. $\sigma\sim_{\uparrow}$ in our definition (Definition~\ref{def:up_random_walk}).

The Markov transition matrix of the down walk of \cite{mukherjee2016random} is defined using an absorbing state $\Theta$ as
\begin{equation}
\label{NE:eq:down}
P_{\sigma,\sigma'}=
\begin{cases}
    p & \sigma\neq\Theta\ and\ \sigma'=\sigma \\
    \frac{1-p}{(M-1)(k+1)}&\sigma\neq\Theta\ and\  {\sigma'}\sim_{\downarrow}\overline{\sigma} \\
    \frac{1-p}{(M-1)(k+1)}&\sigma\neq\Theta\ and\  \sigma'\sim_{\downarrow}\sigma \\
    \beta & \sigma\neq\Theta\ and\ \sigma'=\Theta \\
    1 & \sigma=\sigma'=\Theta\\
    0 & otherwise,
\end{cases}
\end{equation}
where $M\coloneqq \max_{\sigma X^{k-1}} \mathrm{deg}(\sigma) $ is the maximal $(k-1)$-degree. 
This definition is also different from our definition (Definition~\ref{def:down_random_walk}).

Let $p^{\sigma_0}_t(\sigma)$ be a probability of being in $\sigma$ after $t$-steps of random walks starting from $\sigma_0$. 
Then, the expectation process is defined as 
$$
\mathcal{E}^{\sigma_0}_t\coloneqq 
p^{\sigma_0}_t(\sigma) - p^{\sigma_0}_t(\overline{\sigma})
$$
and the normalized expectation process is defines through normalization to this process.

\section{Verification of the promise clique homology problem}
\label{sec:weighted_quantum_walk}

In this section, we provide a verification algorithm for the following problem using quantum walks on simplicial complexes.
\begin{definition}[Promise clique homology~\cite{king2023promise}]
Fix functions $k: \mathbb{N} \rightarrow \mathbb{N}$ and $g: \mathbb{N} \rightarrow[0, \infty)$, with $g(n) \geq 1 / \operatorname{poly}(n), k(n) \leq n$. The input to the problem is a vertex-weighted graph $G$ on $n$ vertices
where the weights of vertices are between $1/\operatorname{poly}(n)$ to $1$. 
We denote the clique complex of $G$ by $X$.

The task is to decide whether:
\begin{itemize}
    \item YES: the $k(n)$-th homology group of $X$ is non-trivial. 
    \item NO: the $k(n)$-th homology group of $X$ is trivial.
\end{itemize}
given a promise that in NO cases, the combinatorial Laplacian $\Delta^k$ has minimum eigenvalue $\lambda_{\min }\left(\Delta^k\right) \geq$ $g(n)$.
\end{definition}

Verification of QMA problems has been important in quantum computational complexity. 
For example, a protocol for a certain QMA-type problem has led to a novel interactive protocol for the verification of quantum computing which is done by untrusted quantum computing devices~\cite{gheorghiu2019verification, fitzsimons2018post}. 
Motivated by potentially providing a new protocol for the verification of quantum computing based on quantum walks, we consider the task of the homology problem in a prover and verifier setting. 

We provide a QMA-verification algorithm based on the harmonic quantum walk on simplicial complexes.
We provide a verification algorithm for the promise clique problem as follows: 

\begin{enumerate}
    \item The prover sends an $n$-qubit state $\ket{w}$ to the verifier (which is supposed to in $\ker(\Delta_k)$ in the YES case). 
    \item    The verifier performs the block-measurement of Lemma~\ref{lemma:block_measurement} with respect to the projected unitary encoding of the projector onto $\ker(\Delta_k)$ of Theorem~\ref{thm:main} based on the harmonic quantum walk. 
    \item Finally, if the measurement outcome is $1$,  the verifier outputs YES; otherwise, outputs NO.
\end{enumerate}

In order to perform the block-measurement, we need to know the spectral gap of the combinatorial Laplacian. However, we do not know the spectral gap in the YES case. Instead, we perform the block-measurement that corresponds to the spectral gap of  $g(n)/2$.

\begin{theorem}
\label{thm:promise_clique_homology}
    The above algorithm can efficiently verify the promise clique homology problem with completeness and soundness. 
\end{theorem}

\begin{proof}

{

From 
Theorem~\ref{thm:main}, 
the approximate kernel projector $\tilde{\Pi}_{ker}$ satisfies 
$$
|\tilde{\Pi}_{\mathcal{H}_k}-\mathrm{Proj}(\mathcal{H}_k)| \leq \epsilon.
$$
According to Lemma~\ref{lemma:block_measurement}, we can construct a quantum channel $\tilde{\Lambda}_{\Pi_{\mathcal{H}_k}}$ which is approximately equal to a channel $\Lambda_{\Pi_{\mathcal{H}_k}}$ that maps as 
    $$
    \ket{0}\otimes \ket{\psi} \rightarrow
    \ket{1}\otimes \Pi_{\mathcal{H}_k} \ket{\psi} + \ket{0} \otimes (I-\Pi_{\mathcal{H}_k}) \ket{\psi}, 
    $$
such that $\| \Lambda_{\Pi_{\mathcal{H}_k}} - \tilde{\Lambda}_{\Pi_{\mathcal{H}_k}} \|_\diamond \in \mathcal{O}(\epsilon)$. 
What we can actually implement is a projector that is close to the projector onto the space spanned by eigenstates of $\Delta_k$ with energy below $g(n)/2$.

Therefore, in the NO case, the probability of outputting $1$ is at most $\mathcal{O}(\epsilon)$ for any state provided by the prover. This is because there are no eigenstates of $\Delta_k$ with eigenvalues below $g(n)$, and the gap between $g(n)/2$ and $g(n)$ is at least inverse-polynomially large.
Similarly, in the YES case, for a state $\ket{w}\in \ker(\Delta_k)$ provided by an honest prover, the probability of outputting $1$ is at least $1-\mathcal{O}(\epsilon)$.

By choosing $\epsilon \in O(1/\exp(n))$, the procedure remains efficient. This concludes the proof.
}
\end{proof}

\end{document}

%% file: clique_construction.tex
In this section, we show efficient constructions of up, down, and harmonic quantum walk unitaries for clique complexes.
When we use $\tilde{\mathcal{O}}$, it hides factors in $polylog(n)$, where $n$ is the number of vertices. 
In the following, we assume that quantum access to graph adjacency matrices and vertex-weighting functions can be done in at most linear time in $n$. 
We will prove the following theorem.

\begin{theorem}
\label{theorem:clique}
Let $X$ be a vertex-weighted clique complex over $n$ vertices and $P^{up}$, $P^{down}$, $P$ be the Markov transition matrices of Definitions~\ref{def:up_random_walk},~\ref{def:down_random_walk},~\ref{def:harmonic_random_walk}. 
Suppose that $(P^{up})_{\sigma,\sigma'}$, $(P^{down})_{\sigma,\sigma'}$, $(P)_{\sigma,\sigma'}$ can be represented with $t$-bits with $t\in \mathcal{O}( \log (1/\epsilon) + \log n )$ for every $\sigma,\sigma' \in S_k$ and we are given the normalization factors $K^{up}$, $K^{down}$ and $K$. 
Then we can 
construct $U^{\uparrow\downarrow}_k$, $U^{\downarrow\uparrow}_k$, $U^h_k$ such that 
it holds for all $\sigma,\sigma' \in X^\pm_k$ that 
\begin{itemize}
    \item $\left|(P^{up})_{\sigma,\sigma'} - \bra{\sigma'}\bra{0^{n+2}}\bra{0^m,0^m}
(U^{\uparrow\downarrow}_k)^\dagger\cdot \mathrm{SWAP}\cdot U^{\uparrow\downarrow}_k
\ket{\sigma}\ket{0^{n+2}}\ket{0^m,0^m}\right|
\leq \epsilon$,
    \item $\left| (P^{down})_{\sigma,\sigma'}-\bra{\sigma'}\bra{0^{n+2}}\bra{0^m,0^m}
(U^{\downarrow\uparrow}_k)^\dagger\cdot \mathrm{SWAP}\cdot U^{\downarrow\uparrow}_k
\ket{\sigma}\ket{0^{n+2}}\ket{0^m,0^m}\right|
\leq \epsilon$,
\item $\left| 
P_{\sigma,\sigma'}-\bra{\sigma'}\bra{0^{n+2}}\bra{0^m,0^m}
(U^h_k)^\dagger\cdot \mathrm{SWAP}\cdot U^h_k
\ket{\sigma}\ket{0^{n+2}}\ket{0^m,0^m}\right| \leq \epsilon$,
\end{itemize}
with $\tilde{\mathcal{O}}(n^2\log(1/\epsilon))$ gates and $m=\mathcal{O}(\log n)$, 
where the $\mathrm{SWAP}$ operator swaps between the first and the second $n+2$-qubit registers and the two $m$-qubit registers, respectively. 
\end{theorem}

{
An immediate consequence of Theorem~\ref{theorem:clique} is that the transition operators $P^{up}$, $P^{down}$, and $P$ admit block-encodings constructed from the unitaries $U_k^{\uparrow\downarrow}$, $U_k^{\downarrow\uparrow}$, and $U_k^h$, respectively. We summarize this in the following corollary.

\begin{corollary}[Block-encoding of transition operators]\label{corollary:clique_encoding}

Let
\[
\widetilde U_k^{up} := (U_k^{\uparrow\downarrow})^\dagger \, \mathrm{SWAP}\, U_k^{\uparrow\downarrow}, \quad
\widetilde U_k^{down} := (U_k^{\downarrow\uparrow})^\dagger \, \mathrm{SWAP}\, U_k^{\downarrow\uparrow}, \quad
\widetilde U_k^{h} := (U_k^{h})^\dagger \, \mathrm{SWAP}\, U_k^{h}.
\]
Then $\widetilde U_k^{up}$, $\widetilde U_k^{down}$, and $\widetilde U_k^{h}$ are $\Pi_k^\pm$-projected $(1,\alpha,\varepsilon)$-unitary encodings of $\Pi_k^\pm P^{up}\Pi_k^\pm$, $\Pi_k^\pm P^{down}\Pi_k^\pm$, and $\Pi_k^\pm P \Pi_k^\pm$, respectively, where $\alpha=n+2+2m$. In particular,
\[
\left\|
\Pi_k^\pm P^{up}\Pi_k^\pm
-
(\Pi_k^\pm \otimes \langle 0^\alpha|)\,
\widetilde U_k^{up}\,
(\Pi_k^\pm \otimes |0^\alpha\rangle)
\right\|
\le \varepsilon,
\]
and similarly for $P^{down}$ and $P$, where $\|\cdot\|$ denotes the operator norm.
Moreover, these encodings can be implemented with
\[
t \in \mathcal{O}\left(\log\left(1/\varepsilon\right)+(k+2)\log n\right)
\]
additional workspace qubits.
\end{corollary}

\begin{proof}
We focus on the case of $P^{up}$; the arguments for $P^{down}$ and $P$ are analogous. Let
\[
E := \Pi_k^\pm P^{up}\Pi_k^\pm 
- (\Pi_k^\pm \otimes \langle 0^\alpha|)\,\widetilde U_k^{up}\,(\Pi_k^\pm \otimes |0^\alpha\rangle).
\]
By Theorem~\ref{theorem:clique}, the entries of $E$ satisfy $|E_{\sigma,\sigma'}|\le \epsilon$ for all $\sigma,\sigma'\in X_k^\pm$ which implies
$\|E\|_\infty \le |X_k^\pm|\epsilon$ and $\|E\|_1 \le |X_k^\pm|\epsilon$. 
Using the standard inequality $\|E\| \le \sqrt{\|E\|_1\|E\|_\infty}$, we obtain
\[
\|E\| \le |X_k^\pm|\epsilon.
\]
Therefore, by choosing the entry-wise precision in Theorem~\ref{theorem:clique} as $\epsilon=\varepsilon/|X_k^\pm|$, we get
\[
\|E\| \le \varepsilon.
\]
Substituting this choice into the bound for $t$ in Theorem~\ref{theorem:clique} yields
\[
t \in \mathcal{O}\!\left(\log\!\left(\frac{|X_k^\pm|}{\varepsilon}\right)+\log n\right)
\in
\mathcal{O}\left(\log\left(1/\varepsilon\right)+(k+2)\log n\right),
\]
where we use $|X_k^\pm|\in \mathcal{O}(n^{k+1})$.
\end{proof}
}

In Section~\ref{sec:clique:notations}, we introduce notations that we use in the proof. 
In Section~\ref{sec:clique_construction_up}, we give the proof of the first part of the theorem (construction of $U^{\uparrow\downarrow}_k$). 
In Section~\ref{sec:clique_construction_down}, we give the proof of the second part of the theorem (construction of $U^{\downarrow\uparrow}_k$). 
In Section~\ref{sec:clique_construction_down}, we give the proof of the third part of the theorem (construction of $U^h_k$).

\subsection{Notations}
\label{sec:clique:notations}

Before giving the proofs of theorems, we introduce several notations for the proofs. 
We denote $f(\sigma)=0$ if $\sigma\in X^\pm$ and $f_k(\sigma)=1$ if $\sigma\notin X^\pm$.
We denote the state representing an oriented $k$-simplex $\sigma\in X^\pm_k$ by $\ket{\sigma}=\ket{x_{\sigma},s(\sigma),0}$ where $x_\sigma
$ is the Hamming weight $k+1$ $n$-bit string that represents the unoriented simplex that corresponds to $\sigma$ and  
$$
s(\sigma)= 
\begin{cases}
    0 & if\ \sigma \in X^+_k, \\
    1 & if\ \sigma \in X^-_k. 
\end{cases}
$$
For $\sigma\in X^\pm_k$, we define 
\begin{itemize}
\item $i(\sigma)\in [n]$: the position of $i$-th 0 of $x_{\sigma}$.  ($i\in [n-k-1]$.)
\item $\tilde{i}(\sigma) \in [n]$: the position of $i$-th 1 of $x_{\sigma}$. ($i\in [k+1]$.)
\end{itemize}
Then, we introduce
\begin{itemize}
    \item $x_\sigma^{\uparrow}(i)\coloneqq x_\sigma\oplus 0^{i(\sigma)-1}10^{n-i(\sigma)}$ for $i\in [n-k-1]$.
    \item $x_\sigma^{\downarrow}(i)\coloneqq x_\sigma\oplus 0^{\tilde{i}(\sigma)-1}10^{n-\tilde{i}(\sigma)}$ for $i=1,2,...,k+1$.
\end{itemize}
These are the descriptions of the unoriented simplices obtained by adding a vertex $i(\sigma)$ to $\sigma$ and removing a vertex $\tilde{i}(\sigma)$, respectively. 
Note that $x_\sigma^{\downarrow}(i) \in X_{k-1}$ always holds but $x_\sigma^{\uparrow}(i) \in X_{k+1}$ does not always holds. 
Then, we define
\begin{itemize}
\item Parity of $i$-th 0 of $\sigma$: $P(i,\sigma)\coloneqq \bigoplus_{j=1}^{i(\sigma)}x_{\sigma_j}$.
\item Parity of $i$: $\tilde{P}(i)=0$ if $i$ is odd and $\tilde{P}(i)=1$ if $i$ is even.
\end{itemize}
Using these notations, we can define
\begin{itemize}
\item $\sigma^\uparrow(i)$: $k+1$-simplex corresponding to $x_\sigma^{\uparrow}(i)$ with orientation $P(i,\sigma)$.
\item $\sigma^\downarrow(i)$: $k-1$-simplex corresponding to $x^\downarrow_\sigma(i)$ with orientation $\tilde{P}(i)$.  
\end{itemize}
These are cofaces and faces of $\sigma$ with the induced orientations (as long as $x_\sigma^{\uparrow}(i) \in X_{k+1}$). 
Finally, we introduce
\begin{itemize}
\item $x_{\sigma}^{\uparrow\downarrow}(i,j)\coloneqq 
x_\sigma^{\uparrow}(i)
\oplus 0^{\tilde{j}(\sigma^\uparrow(i))-1}10^{n-\tilde{j}(\sigma^\uparrow(i))}
$.
\item $x_\sigma^{\downarrow\uparrow}(i,j)\coloneqq x_\sigma^{\downarrow}(i)
\oplus 0^{j(\sigma^\downarrow (i))-1}10^{n-j(\sigma^\downarrow (i))}
$.
\end{itemize}
The first $n$-bit string represents the unoriented simplex obtained by removing a vertex $\tilde{j}(\sigma^\uparrow(i))$ from $\sigma^\uparrow(i)$. 
The second $n$-bit string represents the unoriented simplex obtained by adding a vertex $j(\sigma^\downarrow (i))$ to $\sigma^\downarrow(i)$.
Corresponding to the removed and added vertices, we define the following functions of the product of weights. 
\begin{itemize}
\item     $w_\sigma^{\uparrow\downarrow}(i,j)\coloneqq w(i(\sigma))w(\tilde{j}(\sigma^\uparrow(i)))$ for all $i\in [n-k-1],j\in[k+2] $ .
\item $w_\sigma^{\downarrow\uparrow}(i,j)\coloneqq w(\tilde{i}(\sigma))w({j}(\sigma^\uparrow(i)))$.

\end{itemize}

\subsection{Proof of the first part of Theorem~\ref{theorem:clique}}
\label{sec:clique_construction_up}

We first provide the construction of $U^{\uparrow\downarrow}_k$. 
The construction is given using a quantum membership unitary $M$ that acts
for all $b\in\{0,1\}$ and $\sigma \in \{0,1\}^{n+2}$ as follows: 
$$M\ket{\sigma}\ket{b}=\ket{\sigma}\ket{b\oplus f(\sigma)},$$
where $f(\sigma)=0$ if $\sigma\in X^\pm$ and $f(\sigma)=1$ if $\sigma\notin X^\pm$.

Let $U^{\uparrow\downarrow}_k$ be a unitary constructed in the following steps.  
\begin{enumerate}
    \item For $\ket{\sigma}$, compute the positions of $0$'s of $x_\sigma$ in an ancilla register:
    $$
    \bigotimes_{i\in [n-k-1] }\ket{i(\sigma)}.
    $$
    \item Further compute $\tilde{j}(\sigma^\uparrow(i))$ for $j\in [k+2]$ to create the following ancilla state:
    $$
    \bigotimes_{i\in [n-k-1] ,j\in[k+2]}\ket{i(\sigma)} \ket{\tilde{j}(\sigma^\uparrow(i))}.
    $$

    \item Compute the following  weights using the above ancilla state in another ancilla register:
    $$\bigotimes_{i\in [n-k-1],j\in[k+2] } \ket{w_\sigma^{\uparrow\downarrow}(i,j)/K^{up}},$$
    where $w_\sigma^{\uparrow\downarrow}(i,j)=w(i(\sigma))w(\tilde{j}(\sigma^\uparrow(i)))$.
    In this step, we have used two calls to the vertex-weighting functions.
    \item Then, following~\cite{chiang2009efficient}, approximately prepare the following state for all initial states $\ket{\sigma}\ket{0^{n+2}}= \ket{x_\sigma,s(\sigma),0}\ket{0^{n+2}}$ where $\sigma \in X^\pm_k$:
    $$
    \ket{\sigma}\ket{0^{n+2}}\sum_{i\in [n-k-1],j\in[k+2]}
    \sqrt{\frac{w_\sigma^{\uparrow\downarrow}(i,j)}{K^{up}} }\ket{i,j}. 
    $$
    Although we do not prepare this state exactly, we use this exact expression in the following for simplicity. We give a complexity and error analysis later. 
    \item  Create the following state using  $
    \bigotimes_{i\in [n-k-1] ,j\in[k+2]}\ket{i(\sigma)} \ket{\tilde{j}(\sigma^\uparrow(i))}
    $:
    $$
    \ket{\sigma}\ket{0^{n+2}}\sum_{i\in [n-k-1],j\in[k+2]}
    \sqrt{\frac{w_\sigma^{\uparrow\downarrow}(i,j)}{K^{up}} }\ket{i(\sigma),\tilde{j}(\sigma^\uparrow(i))}.
    $$
    \item Then, create the following state:
    $$
    \sum_{\substack{i\in [n-k-1],\\j\in[k+2]}}
    \sqrt{\frac{w_\sigma^{\uparrow\downarrow}(i,j)}{K^{up}} }
\ket{\sigma}\ket{x_{\sigma}^{\uparrow\downarrow}(i,j),s(\sigma)\oplus P(i,\sigma)\oplus \tilde{P}(j,\sigma^\uparrow(i)),f(\sigma^\uparrow(i))}
    \ket{i(\sigma),\tilde{j}(\sigma^\uparrow(i))}.
    $$
This computation is based on the following sub-steps:
\begin{enumerate}
    \item Recall that $x_{\sigma}^{\uparrow\downarrow}(i,j)= x_\sigma\oplus 0^{i(\sigma)-1}10^{n-i(\sigma)}\oplus 0^{\tilde{j}(\sigma^\uparrow(i))-1}10^{n-\tilde{j}(\sigma^\uparrow(i))}$. 
    Therefore, $\ket{x_{\sigma}^{\uparrow\downarrow}(i,j)}$ can be computed by copying $x_\sigma$ from $\ket{\sigma}$ and applying $X$-gates to $i(\sigma)$-th qubit and $\tilde{j}(\sigma^\uparrow(i))$-th qubit conditioned on that the state of the last register is $\ket{i(\sigma),\tilde{j}(\sigma^\uparrow(i))}$.
    \item The state $\ket{s(\sigma)\oplus P(i,\sigma)\oplus \tilde{P}(j,\sigma^\uparrow(i))}$ can be prepared by copying $s(\sigma)$ from $\ket{\sigma}$ and computing the ``parity'' of $\ket{i(\sigma),\tilde{j}(\sigma^\uparrow(i))}$. 
    \item The state $\ket{f(\sigma^\uparrow(i))}$ can be computed by applying the membership unitary $M$ to the state $\ket{x_\sigma\oplus 0^{i(\sigma)-1}10^{n-i(\sigma)}}$ (which appear in the intermediate of step (a)).
\end{enumerate}

    \item Conditioned that $f(\sigma^\uparrow(i))=1$, uncompute the above step (this is to move to the absorbing state if $\sigma^\uparrow(i)\notin X^\pm_{k+1}$). 
    \item 
    Uncompute the following ancilla registers:
 $$
    \bigotimes_{i\in [n-k-1] ,j\in[k+2]}\ket{i(\sigma)} \ket{\tilde{j}(\sigma^\uparrow(i))}.
    $$
    $$\bigotimes_{i\in [n-k-1],j\in[k+2] } \ket{w_\sigma^{\uparrow\downarrow}(i,j)/K^{up}}.$$
\end{enumerate}

Note that in the above construction, we only used one query to the membership oracle and $\mathcal{O}(n)$ query to the weight functions.

For all $\sigma\in X^\pm_k$, $U^{\uparrow\downarrow}_k$ approximately acts as follows: 

\begin{align*}
&U^{\uparrow\downarrow}_k\ket{\sigma}\ket{0^{n+2}}\ket{0^m,0^m}\\
&\simeq
   \sum_{i:f(\sigma^\uparrow(i))=0,j\in[k+2]}
    \sqrt{\frac{w_\sigma^{\uparrow\downarrow}(i,j)}{K^{up}} }
\ket{\sigma}\ket{x_{\sigma}^{\uparrow\downarrow}(i,j),s(\sigma)\oplus P(i,\sigma)\oplus \tilde{P}(j,\sigma^\uparrow(i)),0}
    \ket{i(\sigma),\tilde{j}(\sigma^\uparrow(i))}\\
&\ \ \ \ +
 \sum_{i:f(\sigma^\uparrow(i))=1,j\in[k+2]}
    \sqrt{\frac{w_\sigma^{\uparrow\downarrow}(i,j)}{K^{up}} }
\ket{\sigma}\ket{\Theta}
   \ket{i(\sigma),\tilde{j}(\sigma^\uparrow(i))}.
\end{align*}

The state $\ket{x_{\sigma}^{\uparrow\downarrow}(i,j),s(\sigma)\oplus P(i,\sigma)\oplus \tilde{P}(j,\sigma^\uparrow(i)),0}$ can be seen as the representation of  simplex with vertices $ x_{\sigma}^{\uparrow\downarrow}(i,j)$ and orientation $s(\sigma)\oplus P(i,\sigma)\oplus \tilde{P}(j,\sigma^\uparrow(i))$. Let us denote such a simplex as $\sigma^{\uparrow\downarrow}(i,j)$. 

For $i\in[n-k-1]$ such that $f(\sigma^\uparrow(i))=0$ (i.e., $\sigma^\uparrow(i)\in X^\pm_{k+1}$), we can show that  $\sigma\sim_{\uparrow}\sigma^{\uparrow\downarrow}(i,j)$  
for any $j\in [k+2]$. This is because $x^\uparrow(i(\sigma))$  is the common upper simplex of them and the orientation induced by them $s(\sigma)\oplus P(i,\sigma)$.

Therefore, 

\begin{align*}
&U^{\uparrow\downarrow}_k\ket{\sigma}\ket{0^{n+2}}\ket{0^m,0^m}\\&\simeq
    \sum_{i:f(\sigma^\uparrow(i))=0,j\in[k+2]}
    \sqrt{\frac{w_\sigma^{\uparrow\downarrow}(i,j)}{K^{up}} }
\ket{\sigma}\ket{\sigma^{\uparrow\downarrow}(i,j)}
    \ket{i(\sigma),\tilde{j}(\sigma^\uparrow(i))}
    \\&\ \ \ \ +
    \sum_{i:f(\sigma^\uparrow(i))=1,j\in[k+2]}
    \sqrt{\frac{w_\sigma^{\uparrow\downarrow}(i,j)}{K^{up}} }
\ket{\sigma}\ket{\Theta}
    \ket{i(\sigma),\tilde{j}(\sigma^\uparrow(i))}.
    \end{align*}

Finally, it can be seen that for all $\sigma,\sigma'\in X^\pm_k$,
\begin{align*}
&\bra{\sigma'}\bra{0^{n+2}}\bra{0^m,0^m}
(U^{\uparrow\downarrow}_k)^\dagger\cdot \mathrm{SWAP}\cdot U^{\uparrow\downarrow}_k
\ket{\sigma}\ket{0^{n+2}}\ket{0^m,0^m}
\\\simeq &
\left(\sum_{\substack{i':f(\sigma'^\uparrow(i'))=0,\\j'\in[k+2]}}
    \sqrt{\frac{w_{\sigma'}^{\uparrow\downarrow}(i',j')}{K^{up}} }
\bra{\sigma'}\bra{\sigma'^{\uparrow\downarrow}(i',j')}
    \bra{i'(\sigma'),\tilde{j'}(\sigma'^\uparrow(i'))} 
\right)
    \\&
    \hspace{3cm} \times\left(\sum_{\substack{i:f(\sigma^\uparrow(i))=0,\\j\in[k+2]}}
    \sqrt{\frac{w_\sigma^{\uparrow\downarrow}(i,j)}{K^{up}} }
\ket{\sigma^{\uparrow\downarrow}(i,j)}\ket{\sigma}
    \ket{\tilde{j}(\sigma^\uparrow(i)),i(\sigma)} 
\right)
    \\
    &=
    \begin{cases}
\frac{w(v_\sigma)w(v_{\sigma'})}{K^{up}} & if\ \sigma\sim_\uparrow \sigma',\\
\frac{\sum_{u\in up(\sigma)}w(u)^2}{K^{up}} & if\ \sigma=\sigma',\\
0 & otherwise,
    \end{cases}
\end{align*}
as desired.

\paragraph{Complexity and error analysis.}
In the construction of $U^{\uparrow\downarrow}_k$, we use $\mathcal{O}(n)$ calls to the vertex-weighting function in step 3, and one call to the membership function in step 6. 
We assumed that the cost of calling the vertex-weighting function is $\mathcal{O}(n)$. 
The cost of implementing the membership function for clique complex is $\mathcal{O}(n^2)$ because it suffices to check the connection of all the pairs of vertices. The number of pair of vertices for $k$-simplices is $\binom{k+1}{2}$ and $k\leq n-1$.
The error in the construction comes from the creation of the superposition in step 4. As it is shown in~\cite{chiang2009efficient}, the gate complexity for making the error of the amplitude in each of the superposition upper bounded by $\mathcal{O}(\sqrt{\epsilon})$ is $\tilde{\mathcal{O}}(n\log(1/\epsilon))$. 
Other computations in the construction can be done with $\mathcal{O}(n^2)$ size circuit with classical reversible gates. 
This concludes the proof of the first part of Theorem~\ref{theorem:clique}.

\subsection{Proof of the second part of Theorem~\ref{theorem:clique}}
\label{sec:clique_construction_down}

{
We first give a construction of $U^{\downarrow\uparrow}_k$. 
Let $U^{\downarrow\uparrow}_k$ be a unitary that is constructed in the following steps: }

\begin{enumerate}
    \item For $\ket{\sigma}$, compute the positions of $1$'s of $x_\sigma$ in an registers:
    $$
    \bigotimes_{i\in [k+1] }\ket{\tilde{i}(\sigma)}.
    $$
    \item Compute $j(\sigma^\downarrow(i))$ for $j\in [n-k]$ to create the following ancilla state:
    $$
    \bigotimes_{i\in [k+1] , j\in [n-k]}\ket{\tilde{i}(\sigma)}\ket{j(\sigma^\downarrow(i))}.
    $$
    \item Compute the following  weights in another ancilla register:
    $$\bigotimes_{i\in [k+1],j\in[n-k] } \ket{w_\sigma^{\downarrow\uparrow}(i,j)/K^{down                              }},$$
    where $w_\sigma^{\downarrow\uparrow}(i,j)\coloneqq w(\tilde{i}(\sigma))w({j}(\sigma^\uparrow(i)))$.
    \item 
    Then, following~\cite{chiang2009efficient}, approximately prepare the following state for all initial states $\ket{\sigma}\ket{0^{n+2}}$ where $\sigma \in X^\pm_k$:
    $$
    \ket{x_\sigma,s(\sigma),0}\ket{0^{n+2}}\sum_{i\in [k+1],j\in[n-k]}
    \sqrt{\frac{w_\sigma^{\downarrow\uparrow}(i,j)}{K^{down}} }\ket{i,j}.
    $$
    \item     Create the following state using the information in ancilla registers:
    $$
    \ket{x_\sigma,s(\sigma),0}\ket{0^{n+2}}\sum_{i\in [k+1],j\in[n-k]}
    \sqrt{\frac{w_\sigma^{\uparrow\downarrow}(i,j)}{K^{up}} }\ket{\tilde{i}(\sigma),j(\sigma^\downarrow(i))}.
    $$
    \item Create the following state:
    $$
    \sum_{\substack{i\in [k+1],\\j\in[n-k]}}
    \sqrt{\frac{w_\sigma^{\downarrow\uparrow}(i,j)}{K^{down}} }
\ket{x_\sigma,s(\sigma),0}\ket{x_\sigma^{\downarrow\uparrow}(i,j),s(\sigma)\oplus \tilde{P}(i)\oplus {P}(j,\sigma^\downarrow(i)),f(x_\sigma^{\downarrow\uparrow}(i,j))} \ket{\tilde{i}(\sigma),j(\sigma^\downarrow(i))}.$$
The computation is similar to the step 6 of the previous subsection.
    \item Conditioned that $f(x_\sigma^{\downarrow\uparrow}(i,j))=1$ (i.e., $x_\sigma^{\downarrow\uparrow}(i,j)\in X_k$), uncompute the above step. 
    \item 
    Uncompute the following ancilla registers:
    $$
    \bigotimes_{i\in [n-k-1] ,j\in[k+2]}\ket{i(\sigma)} \ket{\tilde{j}(\sigma^\uparrow(i))},
    $$
   $$\bigotimes_{i\in [k+1],j\in[n-k] } \ket{w_\sigma^{\downarrow\uparrow}(i,j)/K^{down                              }}.$$
\end{enumerate}

The computation is similar to the first part of the proof of Theorem~\ref{theorem:clique}.
    In this case, it holds that 
    \begin{align*}
&U^{\downarrow\uparrow}_k\ket{\sigma}\ket{0^{n+2}}\ket{0^m,0^m}\\&\simeq 
    \sum_{\substack{i\in [k+1],\\j:f(x_\sigma^{\downarrow\uparrow}(i,j))=0}}
    \sqrt{\frac{w_\sigma^{\downarrow\uparrow}(i,j)}{K^{up}} }
\ket{\sigma}\ket{\sigma^{\downarrow\uparrow}(i,j)}
    \ket{\tilde{i}(\sigma),{j}(\sigma^\downarrow(i))}
    \\&+
    \sum_{\substack{i\in [k+1],\\j:f(x_\sigma^{\downarrow\uparrow}(i,j))=1}}
    \sqrt{\frac{w_\sigma^{\downarrow\uparrow}(i,j)}{K^{up}} }
\ket{\sigma}\ket{\Theta}
    \ket{\tilde{i}(\sigma),{j}(\sigma^\downarrow(i))},
    \end{align*}
    where in the first term, 
    $\ket{\sigma^{\downarrow\uparrow}(i,j)}=
    \ket{x_\sigma^{\downarrow\uparrow}(i,j),s(\sigma)\oplus \tilde{P}(i)\oplus {P}(j,\sigma^\downarrow(i)),0}$
    and $\sigma\sim_{\downarrow} \sigma^{\downarrow\uparrow}(i,j)$. 
    This is because due to the condition that $j:f(x_\sigma^{\downarrow\uparrow}(i,j))=0$, $\sigma$ and $ \sigma^{\downarrow\uparrow}(i,j)$ shares a same lower simplex and the orientation of the lower simplex induced by both $\sigma$ and $\sigma^{\downarrow\uparrow}(i,j)$ is $s(\sigma)\oplus \tilde{P}(i)$.

    The relation between the error and the complexity is similar to the case of the first part of Theorem~\ref{theorem:clique}. 
    Then, by a similar calculation as in the first part of the proof of Theorem~\ref{theorem:clique}, we can conclude the second part of the proof of the theorem.

\subsection{Proof of the third part of Theorem~\ref{theorem:clique}}
\label{sec:clique_construction_harmonic}

First, we show the following claim:

\begin{claim}
\label{claim:parity}
For $\sigma,\sigma' \in X^\pm_k$ s.t. $\sigma\sim_{\downarrow}\sigma'$, let $v_\sigma$ and $v_{\sigma'}$ be the vertices to be removed from $\sigma$ and $\sigma'$ to obtain the common lower simplex, respectively. It holds that $s(\sigma')=s(\sigma)+ \sum_{i\in (v_\sigma, v_{\sigma'})} (x_\sigma)_i \ (\mathrm{mod} 2)$ . Here, $(v_\sigma, v_{\sigma'})$ are set of integers between $v_{\sigma}+1$ and $v_{\sigma'}-1$ if $v_{\sigma}<v_{\sigma'}$ and integers between $v_{\sigma'}+1$ and $v_{\sigma}-1$ if $v_{\sigma'}<v_{\sigma}$ .

\end{claim}

\begin{proof}
Suppose that $v_{\sigma}<v_{\sigma'}$. Then, the lower simplex obtained by removing $v_\sigma$ from $\sigma$ has an orientation $s(\sigma)+\sum_{i\in [1,v_\sigma]} (x_\sigma)_i+1$ . Therefore, the orientation of $\sigma'$ is 
$$
\left(s(\sigma)+\sum_{i\in [1,v_\sigma]} (x_\sigma)_i+1\right)+\sum_{j\in [1,v_{\sigma'})}(x_\sigma)_i+1 
=
s(\sigma)+ \sum_{i\in (v_\sigma, v_{\sigma'})} (x_\sigma)_i \ (\mathrm{mod} 2).
$$
Here, $[1,v_{\sigma'})=\{1,2,...,v_{\sigma'}-1\}$ and $(x_\sigma)_i$ is the $i$-th bit of $x_\sigma$.
Next, suppose that $v_{\sigma'}<v_{\sigma}$. Then, the orientation of $\sigma'$ is 
$$\left(s(\sigma)+\sum_{i\in [1,v_\sigma]} (x_\sigma)_i+1\right)+\sum_{j\in [1,v_{\sigma'}]}(x_\sigma)_i=s(\sigma)+ \sum_{i\in (v_\sigma', v_{\sigma})} (x_\sigma)_i \ (\mathrm{mod} 2).$$
This concludes the proof of the claim.

\end{proof}

Then, we show the construction of $U^h$. Let $U^h$
be the unitary constructed in the following steps:
\begin{enumerate}
    \item Compute the positions of 0's and 1's of $\sigma$:
$$\bigotimes_{i\in [n-k-1] }\ket{i(\sigma)}\bigotimes_{j\in [k+1] }\ket{\tilde{j}(\sigma)}.$$
    \item     Compute the following  weights in another ancilla register:
    $$\bigotimes_{i\in [n-k-1],j\in[k+1] } \Ket{\frac{w(i(\sigma))w(\tilde{j}(\sigma))}{K}}
\Ket{\frac{w(i(\sigma))^2}{K}} \Ket{\frac{w(\tilde{j}(\sigma))^2}{K}}.
$$

\item Approximately create the following superposition following~\cite{chiang2009efficient}:
\begin{align*}
  & \sum_{\substack{i\in [n-k-1],\\j\in [k+1]}} \sqrt{\frac{w(i(\sigma))w(\tilde{j}(\sigma))}{K}}\ket{i(\sigma)}\ket{\tilde{j}(\sigma)}\ket{0} 
+
\sum_{i\in [n-k-1]} \sqrt{\frac{w(i(\sigma))^2}{{K}}}\ket{i(\sigma)}\ket{i(\sigma)}\ket{1} \\&
+
\sum_{j\in [k+1]} \sqrt{\frac{w(\tilde{j}(\sigma))^2}{{K}}}\ket{\tilde{j}(\sigma)}\ket{\tilde{j}(\sigma)}\ket{2}
=: \ket{\phi(\sigma)} .
\end{align*}
The first term corresponds to the superposition for adjacency, the second term corresponds to the up degree, and the third term corresponds to a lower degree of $\sigma$. 

\item Prepare $\ket{\sigma}\ket{0^{n+2}}\ket{\phi(\sigma)}$.

\item Define 
\begin{align*}
&\ket{\psi_{ij}(\sigma)}\coloneqq\\&
\Ket{x_\sigma^{\uparrow}(i) \oplus 0^{\tilde{j}(\sigma)-1}10^{n-\tilde{j}(\sigma)},s(\sigma)+\bigoplus_{m\in (i(\sigma) ,\tilde{j}(\sigma))}(x_\sigma)_m,f(x_\sigma^{\uparrow}(i)\oplus 0^{\tilde{j}(\sigma)-1}10^{n-\tilde{j}(\sigma)})\oplus f(x_\sigma^{\uparrow}(i))\oplus 1}.
\end{align*}
Here, $x_\sigma^{\uparrow}(i) \oplus 0^{\tilde{j}(\sigma)-1}10^{n-\tilde{j}(\sigma)}$ is the unoriented simplex obtained by exchanging between the $j$-th vertex of $\sigma$ and the $i$-th ``empty vertex'' of $\sigma$. 
\item    Create the following state:
\begin{align}
\label{eq:theorem_harm}
&\sum_{\substack{i\in [n-k-1],\\j\in [k+1]}} 
\sqrt{\frac{w(i(\sigma))w(\tilde{j}(\sigma))}{K}}
\ket{\sigma}
\ket{\psi_{ij}(\sigma)}
\ket{i(\sigma)}\ket{\tilde{j}(\sigma)}\ket{0}\\ \nonumber
&
+
\sum_{i\in [n-k-1]} \sqrt{\frac{w(i(\sigma))^2}{{K}}}
\ket{\sigma}
\ket{x_\sigma,s(\sigma),f(\sigma^\uparrow(i(\sigma)))}
\ket{i(\sigma)}\ket{i(\sigma)}\ket{1}\\ \nonumber
&
+
\sum_{j\in [k+1]} \sqrt{\frac{w(\tilde{j}(\sigma))^2}{{K}}}
\ket{\sigma}
\ket{\sigma}
\ket{\tilde{j}(\sigma)}\ket{\tilde{j}(\sigma)}\ket{1}.
\end{align}
\item Uncompute the above step conditioned on that the $n+2$-th qubit of the second register is $\ket{1}$.
\end{enumerate}

Note that in the above construction, we used $\mathcal{O}(1)$ query to the membership oracle and $\mathcal{O}(n)$ query to the weight functions.

In the first term of eq.~\eqref{eq:theorem_harm}, 
observe that $f(x_\sigma^{\uparrow}(i))=0$ implies $f(x_\sigma^{\uparrow}(i)\oplus 0^{\tilde{j}(\sigma)-1}10^{n-\tilde{j}(\sigma)})=0$.  
This is because if $f(x_\sigma^{\uparrow}(i))=0$ i.e., $x_\sigma^{\uparrow}(i)\in X_{k+1}$, 
$x_\sigma^{\uparrow}(i)\oplus 0^{\tilde{j}(\sigma)-1}10^{n-\tilde{j}(\sigma)}$, which is a face of $x_\sigma^{\uparrow}(i)$, is also in $X_k$.
Then,
$$f(x_\sigma^{\uparrow}(i)\oplus 0^{\tilde{j}(\sigma)-1}10^{n-\tilde{j}(\sigma)})\oplus f(x_\sigma^{\uparrow}(i))\oplus 1 =0 $$
means that $f(x_\sigma^{\uparrow}(i))=1$ and $f(x_\sigma^{\uparrow}(i)\oplus 0^{\tilde{j}(\sigma)-1}10^{n-\tilde{j}(\sigma)})=0$. 
Therefore, 
it can be observed that 
$f(x_\sigma^{\uparrow}(i)\oplus 0^{\tilde{j}(\sigma)-1}10^{n-\tilde{j}(\sigma)})\oplus f(x_\sigma^{\uparrow}(i))\oplus 1 =0$ 
means that 
$\ket{\psi_{ij}(\sigma)}$ 
are correct representations of ``lower adjacent simplices of $\sigma$ that are not upper adjacent'' 
because due to Claim~\ref{claim:parity}, 
$\bigoplus_{m\in (i(\sigma) ,\tilde{j}(\sigma))}(x_\sigma)_m$ 
represents the relative orientation of lower adjacent simplices. 
In the second term of eq.~\eqref{eq:theorem_harm}, 
the elements of superposition that satisfy $f(\sigma^\uparrow(i(\sigma)))=0$ 
corresponds to laziness due to the upper degree. 
These arguments also show that the uncomputation in step 7 leads to the absorbing state with the desired condition. 
The relationship between the error and the complexity is similar in the case of the first part of the proof of Theorem~\ref{theorem:clique}. 
Therefore, with a similar calculation, we can conclude the proof of the theorem.

%% file: application.tex
In this section, we apply the derived quantum walk algorithms to resolve the following three problems: estimation of the normalized Betti numbers, estimation of the normalized persistent Betti numbers, and the verification of the promise clique homology problem. 
{
In these applications, we consider a special class of simplices called clique complexes. Clique complexes are closely related to the Vietoris-Rips complexes which are commonly used in the topological data analysis for the point cloud data. 
The Vietoris-Rips complexes are constructed by first connecting the vertices by taking the intersections of balls obtained by fattening the vertices with a certain radius and then computing the cliques of the graph. 
Therefore, clique complexes are an important class of simplicial complexes that can be succinctly described by graphs and are of practical interest. 
}

\subsection{Estimating the normalized Betti numbers}\label{sec:app:Betti_numbers}

In this subsection, we show how to estimate the Betti numbers using harmonic quantum walk on simplicial complexes. 

\begin{theorem}
\label{thm:Betti_numbers}
Given $G=(V,E)$ where $|V|=n$, let $X$ be a clique complex of the (unweighted) graph $G$ and {$U^h$ be a  $\Pi^\pm_k$-projected $(1,\mathcal{O}(n),0)$-unitary encoding of $\Pi^\pm_k P \Pi^\pm_k$} where $P$ is the Markov transition matrix defined in Definition~\ref{def:harmonic_random_walk}.  
    Suppose $X$ satisfies the following promises: 
    \begin{itemize}
        \item We can (classically or quantumly) approximately sample from the uniform distribution over the $k$-simplices of $X$ within a total variation distance (TVD) $\delta$ for any $\delta\in \Omega(1/poly(n))$ in $poly(n)$ time. 
        \item The smallest non-zero eigenvalue $\lambda$ of the $k$-th combinatorial Laplacian $\Delta_k$ of $X$ satisfies $\lambda \in \Omega(1/poly(n))$.
    \end{itemize}
    Then, we can output an estimate for $\beta_k/n_k$ within additive error $\epsilon$ in $poly(n)$ time for any $\epsilon \in \Omega(1/poly(n))$ using 
    $poly(n)$ number of harmonic quantum walk unitary $U^h$,  $(U^h)^\dagger$, and other single-qubit and two-qubit gates, respectively. 
\end{theorem}

\begin{proof}
First, based on the assumption, sample from the uniform distribution over $k$-simplices within TVD error $\delta$.
Then, we use the sample as the input to the quantum circuit and prepare $\tilde{\rho}_k$ such that 
$$
 \| \tilde{\rho}_k-\rho_k \| \leq \delta,
$$
where
$$
\rho_k\coloneqq \frac{1}{n_k}\sum_{\sigma \in X^+_k}\ket{\sigma}\bra{\sigma}. 
$$
Note that 
$$
\rho _k = \frac{1}{n_k}\sum_{\sigma \in X^+_k}\ket{\sigma}\bra{\sigma} 
= \frac{1}{n_k}
\sum_{i=1}^{n_k} \ket{\psi_i}\bra{\psi_i},
$$
where $\{\ket{\psi_i}\}$ are the eigenvectors of the combinatorial Laplacian.

We can easily modify the unitary encoding $U_{\mathcal{H}_k}$ of Theorem~\ref{thm:main} (with error $\epsilon'$) into block-encoding using $\mathrm{C}_{\Pi_k} \mathrm{NOT}$ with an additional ancilla qubit. 
Then we apply the block-measurement of Lemma~\ref{lemma:block_measurement}. 
Let us denote the probability of outputting $1$ by measuring the first qubit as $p_1$.
We have
\begin{align*}
|{p}_1-\beta_k/n_k|
&\leq    \| \tilde{\rho}_k-\rho_k \| + \|\tilde{\Lambda}_\Pi - \Lambda_\Pi\|_\diamond\\
&= \delta + O(\epsilon').
\end{align*}
Then, we estimate the probability $p_1$ by drawing $N=\varepsilon^{-2}$ samples from this quantum circuit. 
Let us denote the number of measurement outcomes $1$ of the $N$ samples as $N_1$. 
We output $N_1/N$ as the estimate of the normalized Betti numbers. 
By the Chernoff-Hoeffding inequality, $N_1/N$ is a correct estimate of $p_1$ within additive error $\varepsilon$ with high probability. 
Therefore, we can estimate $\beta_k/n_k$ within arbitrary additive error $\epsilon \in \Omega(1/poly(n))$ by taking small enough $\delta, \epsilon',\varepsilon$ while keeping the algorithm efficient. 

\end{proof}

\subsection{Estimating the normalized persistent Betti numbers}
\label{sec:app:persistent_Betti_numbers}

In this subsection, we show how to estimate the persistent Betti numbers using quantum walks. 
We show the following theorem: 

\begin{theorem}
\label{thm:persisntent_Betti_numbers}
Given two unweighted graphs $G_i=(V,E_i)$ and $G_j=(V,E_j)$, where $|V|=n$ and $E_i \subseteq E_j$, let $X^i$, $X^j$ be clique complexes of the graph $G_i$ and $G_j$.
Let $U^{down}_i$ be a $\Pi^{\pm,i}_k$-projected 
$(1,\mathcal{O}(n),0)$-unitary encoding of $\Pi_k^{\pm,i} P^{down}_i \Pi_k^{\pm,i}$, where $\Pi_k^{\pm,i}$ is the projector onto the space spanned by $\{\ket{\sigma}\}_{\sigma\in X^{\pm,i}_k}$. 
Similarly, let $U^{up}_j$ be a $\Pi^{\pm,j}_k$-projected $(1,\mathcal{O}(n),0)$-unitary encoding of $\Pi^{\pm,j}_k P^{up}_j \Pi^{\pm,j}_k$ where where $\Pi_k^{\pm,j}$ is the projector onto the space spanned by $\{\ket{\sigma}\}_{\sigma\in X^{\pm,j}_k}$. 
Suppose $X^i$ and $X^j$ satisfy the following promises: 
    \begin{itemize}
        \item We can (classically or quantumly) approximately sample from the uniform distribution over the $k$-simplices of $X^i$ within a total variation distance (TVD) $\delta$ for any $\delta\in \Omega(1/poly(n))$ in $poly(n)$ time. 
        \item The smallest non-zero eigenvalue $\lambda^{i,down}$ of the $k$-th down Laplacian $\Delta_k^{i,down}$ of $X^i$ satisfies $\lambda^{i,down} \in \Omega(1/poly(n))$.
        \item The smallest non-zero eigenvalue $\lambda^{j,up}$ of the $k$-th down Laplacian $\Delta_k^{j,up}$ of $X^j$ satisfies $\lambda^{j,up} \in \Omega(1/poly(n))$.
        \item {
        The singular value $1$ of
        \[\mathrm{Proj}\bigl(\ker(\Delta_k^{i,down})\bigr)\cdot \mathrm{Proj}\bigl(\ker(\Delta_k^{j,up})^{\perp}\bigr)
        \]
        is separated from the rest of the singular spectrum by an inverse-polynomial gap.}
    \end{itemize}
    Then, we can output an estimate for $\beta^{i,j}_k/n_k$ within additive error $\epsilon$ in $poly(n)$ time for any $\epsilon \in \Omega(1/poly(n))$ using 
    $poly(n)$ number of $U^{i,down}$ , $(U^{i,down})^\dagger$, $U^{j,up}$ , $(U^{j,up})^\dagger$ and other single-qubit and two-qubit gates, respectively. 
\end{theorem}

\begin{proof}
This result is inspired by the quantum algorithm of~\cite{mcardle2022streamlined}. 
The quantum algorithm of~\cite{mcardle2022streamlined} is based on the following fact: 
$$
\beta^{i,j}_k = \dim (\ker (\partial_k^i))- \dim (\ker (\partial_k^i) \cap \mathrm{Im}(\partial_{k+1}^j)).
$$
It can be seen that 
$$
\ker(\partial_k^i)=\ker((\partial^i_{k})^*\partial_k^i)=\ker(\Delta_k^{i,down})
= Z_k^i,
$$
and
$$
\mathrm{Im}(\partial_{k+1}^j)
=
\mathrm{Im}(\partial_{k+1}^j(\partial_{k+1}^j)^*)
=
\mathrm{Im}(\Delta_k^{j,up})
=
B_k^j
.
$$
Therefore, 
$$
\beta^{i,j}_k = \dim (\ker (\Delta_k^{i,down}))- \dim (\ker (\Delta_k^{i,down}) \cap (\mathrm{Im}(\Delta_k^{j,up})).
$$
We can therefore estimate $\beta^{i,j}_k/n_k$ with estimating 
$$\frac{\dim (\ker (\Delta_k^{i,down}))}{n_k}\ \ \text{and}\ \ 
\frac{\dim (\ker (\Delta_k^{i,down}) \cap (\mathrm{Im}(\Delta_k^{j,up}))}{n_k}
$$
separately. 
The estimation of the first quantity can be achieved with the unitary encoding of 
$$
\mathrm{Proj}(\ker (\Delta_k^{i,down})),
$$
and the estimation of the second quantity can be achieved with the unitary encoding of 
$$
\mathrm{Proj}\left(\ker (\Delta_k^{i,down}) \cap (\mathrm{Im}(\Delta_k^{j,up})\right).
$$  

{
Next, we estimate
\[
\dim\!\left(\ker(\Delta_k^{i,down}) \cap \ker(\Delta_k^{j,up})^\perp\right)/n_k.
\]
A unitary encoding of
\[
\mathrm{Proj}\left(\ker(\Delta_k^{i,down}) \cap \ker(\Delta_k^{j,up})^\perp\right)
\]
can be implemented as follows. First, we construct a unitary encoding of
\begin{align}\label{eq:prod_proj}
\mathrm{Proj}\!\left(\ker (\Delta_k^{i,down})\right)\cdot
\mathrm{Proj}\!\left( \ker(\Delta_k^{j,up})^\perp\right) 
\end{align}
by composing the unitary encodings~\cite{gilyen2019quantum} of
\[
\mathrm{Proj}\left(\ker(\Delta_k^{i,down})\right)
\quad \text{and} \quad
\mathrm{Proj}\left(\ker(\Delta_k^{j,up})^\perp\right).
\]
Then, following~\cite{mcardle2022streamlined}, we apply QSVT to the resulting unitary. 
The singular value $1$ of the operator in eq.~\eqref{eq:prod_proj} coincides with the subspace
\[
\ker(\Delta_k^{i,down}) \cap \ker(\Delta_k^{j,up})^\perp,
\]
and is separated from the rest of the singular spectrum by an inverse-polynomial gap under the assumption. 
Indeed, a vector lies in the singular value $1$ subspace of the operator in eq.~\eqref{eq:prod_proj} if and only if it is invariant under both projectors, and hence belongs to 
\[
\ker(\Delta_k^{i,down}) \cap \ker(\Delta_k^{j,up})^\perp.
\]
Therefore, QSVT yields a unitary encoding of
\[
\mathrm{Proj}\left(\ker(\Delta_k^{i,down}) \cap \ker(\Delta_k^{j,up})^\perp\right).
\]

Using this encoding, we can estimate
\[
\dim\!\left(\ker(\Delta_k^{i,down}) \cap \ker(\Delta_k^{j,up})^\perp\right)/n_k
\]
in the same manner as in the proof of Theorem~\ref{thm:Betti_numbers}.
The estimation can be performed efficiently provided that
\[
1/\lambda^{i,down},\; 1/\lambda^{j,up} \in poly(n).
\]
 }

\end{proof}